\newif\ifsubmit     %
\newif\ifblind      %
\newif\ifcompact    %
\newif\ifexabs      %
\newif\ifitcs       %

\newif\ifshowabs    %
\ifitcs
  \documentclass[a4paper,UKenglish,cleveref]{lipics-v2021}
  \exabstrue
  \showabstrue
\else
  \documentclass[letterpaper,11pt,pdfa]{article}
  \usepackage{fullpage}
  \ifexabs\else\showabstrue\fi
  
  \usepackage{iftex}
  \ifPDFTeX
    \usepackage[utf8]{inputenc}
    \usepackage[noTeX]{mmap}
    \usepackage[T1]{fontenc}
  \fi
  \ifLuaTeX
    \usepackage{luatex85}
    \usepackage[noTeX]{mmap}
  \fi

\fi

\usepackage{amsmath, amsfonts, amsthm, amssymb, color}
\ifitcs\else
  \usepackage[bookmarks]{hyperref}
  \usepackage[nameinlink]{cleveref}
  \ifexabs
    \newtheorem{theorem}{Theorem}
  \else
    \newtheorem{theorem}{Theorem}[section]
  \fi
  \newtheorem{definition}[theorem]{Definition}
  
  \newtheorem{lemma}[theorem]{Lemma}
  \newtheorem{corollary}[theorem]{Corollary}
  \newtheorem{proposition}[theorem]{Proposition}
  \newtheorem{claim}[theorem]{Claim}
  \newtheorem*{remark*}{Remark}
\fi
\newtheorem{construction}[theorem]{Construction}
\newtheorem{fact}[theorem]{Fact}

\newtheorem*{theorem*}{Theorem}
\newtheorem*{lemma*}{Lemma}

\usepackage{appendix}
\usepackage{algorithm}
\usepackage{algpseudocode}
\usepackage{braket}
\usepackage{comment}
\usepackage{url}
\usepackage{multicol}

\usepackage{pgfplots, tikz}
\pgfplotsset{compat=1.14}

\ifcompact
  \usepackage{enumitem}
  \setlist[description]{noitemsep, topsep=0pt}
  \setlist[enumerate]{noitemsep, topsep=0pt}
  \setlist[itemize]{noitemsep, topsep=0pt}
\fi

\ifitcs
\else
  \usepackage[style=alphabetic,minalphanames=3,maxalphanames=4,maxnames=99,backref=true]{biblatex}
  \renewcommand*{\labelalphaothers}{\textsuperscript{+}}
  \renewcommand*{\multicitedelim}{\addcomma\space}
\fi

\usepackage{soul, xcolor, xparse}
\makeatletter
  \ExplSyntaxOn
    \cs_new:Npn \white_text:n #1
    {
      \fp_set:Nn \l_tmpa_fp {#1 * .01}
      \llap{\textcolor{white}{\the\SOUL@syllable}\hspace{\fp_to_decimal:N \l_tmpa_fp em}}
      \llap{\textcolor{white}{\the\SOUL@syllable}\hspace{-\fp_to_decimal:N \l_tmpa_fp em}}
    }
    \NewDocumentCommand{\whiten}{ m }
    {
      \int_step_function:nnnN {1}{1}{#1} \white_text:n
    }
  \ExplSyntaxOff
  
  \NewDocumentCommand{ \varul }{ D<>{5} O{0.2ex} O{0.1ex} +m } {%
    \begingroup
    \setul{#2}{#3}%
    \def\SOUL@uleverysyllable{%
      \setbox0=\hbox{\the\SOUL@syllable}%
      \ifdim\dp0>\z@
      \SOUL@ulunderline{\phantom{\the\SOUL@syllable}}%
      \whiten{#1}%
      \llap{%
        \the\SOUL@syllable
        \SOUL@setkern\SOUL@charkern
      }%
      \else
      \SOUL@ulunderline{%
        \the\SOUL@syllable
        \SOUL@setkern\SOUL@charkern
      }%
      \fi}%
    \ul{#4}%
    \endgroup
  }
\makeatother

\newcommand{\E}{\mathop{\mathbb{E}}}
\newcommand{\Z}{\mathbb{Z}}
\newcommand{\I}{\mathbb{I}}

\newcommand{\As}{\mathcal{A}}
\newcommand{\Bs}{\mathcal{B}}
\newcommand{\Cs}{\mathcal{C}}

\newcommand{\Ts}{\mathcal{T}}
\newcommand{\Ws}{\mathcal{W}}
\newcommand{\Zs}{\mathcal{Z}}

\newcommand{\Hs}{\mathcal{H}}

\newcommand{\sX}{\mathcal{X}}
\newcommand{\sY}{\mathcal{Y}}
\newcommand{\sR}{\mathcal{R}}

\newcommand{\prg}{{\sf PRG}}

\newcommand{\ver}{{\sf Ver}}

\newcommand{\ans}{{\sf ans}}

\newcommand{\oneoftwo}{\texorpdfstring{$1$-of-$2$}{1-of-2} }
\newcommand{\twooftwo}{\texorpdfstring{$2$-of-$2$}{2-of-2} }
\newcommand{\oneofpoweroftwo}{\texorpdfstring{$1$-of-$2^k$}{1-of-2\^{}k} }
\newcommand{\keygen}{{\sf KeyGen}}
\newcommand{\obligate}{{\sf Obligate}}
\newcommand{\solve}{{\sf Solve}}

\newcommand{\sk}{{\sf sk}}

\newcommand{\pk}{{\sf pk}}

\newcommand{\prpv}{\mathsf{PRPV}}

\newcommand{\prpvrom}{\mathsf{ROPRPV}}

\newcommand{\poly}{{\sf poly}}
\newcommand{\negl}{{\sf negl}}

\DeclareMathOperator{\Tr}{Tr}

\newcommand{\supp}{\mathsf{SUPP}}

\newcommand{\PB}{\mathsf{PB}}
\newcommand{\PC}{\mathsf{PC}}

\ifsubmit
    \newcommand{\luowen}[1]{}
    \newcommand{\qipeng}[1]{}
    \newcommand{\jiahui}[1]{}
\else
    \newcommand{\luowen}[1]{{\color{magenta} Luowen: #1}}
    \newcommand{\qipeng}[1]{{\color{red} Qipeng: #1}}
    \newcommand{\jiahui}[1]{{\color{blue} Jiahui: #1}}
\fi

\newcommand{\td}{{\sf td}}

\title{Beating Classical Impossibility of Position Verification}
\newcommand{\acks}{
  The authors would like to thank Ran Canetti and Shih-Han Hung
  for their helpful discussions.
}
\newcommand{\jiahuifunding}{supported by the NSF and Scott Aaronson's Simons Investigator award.}
\newcommand{\qipengfunding}{supported by the Simons Institute for the Theory of Computing, through a Quantum Postdoctoral Fellowship.}
\newcommand{\luowenfunding}{supported by DARPA under Agreement No. HR00112020023.}

\ifitcs
\author{Jiahui Liu}{Department of Computer Science, University of Texas at Austin, USA}{jiahui@cs.utexas.edu}{}{\jiahuifunding}
\author{Qipeng Liu}{Simons Institute for the Theory of Computing, USA}{qipengliu0@gmail.com}{}{\qipengfunding}
\author{Luowen Qian}{Department of Computer Science, Boston University, USA}{luowenq@bu.edu}{}{\luowenfunding}

\authorrunning{J. Liu, Q. Liu, and L. Qian}

\Copyright{Jiahui Liu, Qipeng Liu, and Luowen Qian}

\ccsdesc[500]{Theory of computation~Cryptographic protocols}
\ccsdesc[300]{Security and privacy~Authorization}
\ccsdesc[300]{Security and privacy~Public key (asymmetric) techniques}
\ccsdesc[100]{Theory of computation~Quantum complexity theory}
\ccsdesc[100]{Theory of computation~Quantum query complexity}

\keywords{cryptographic protocol, position verification, quantum cryptography, proof of quantumness, non-locality}

\relatedversion{}
\relatedversiondetails{Full Version}{https://arxiv.org/abs/2109.07517}

\acknowledgements{\acks}
\nolinenumbers

\EventEditors{Mark Braverman}
\EventNoEds{1}
\EventLongTitle{13th Innovations in Theoretical Computer Science Conference (ITCS 2022)}
\EventShortTitle{ITCS 2022}
\EventAcronym{ITCS}
\EventYear{2022}
\EventDate{January 31--February 3, 2022}
\EventLocation{Berkeley, CA, USA}
\EventLogo{}
\SeriesVolume{215}
\ArticleNo{116}
\else\ifblind\else
\newcommand{\email}[1]{\href{mailto:#1}{\texttt{#1}}}
\author{
  Jiahui Liu\footnote{University of Texas at Austin. Email: \email{jiahui@cs.utexas.edu}}
  \and
  Qipeng Liu\footnote{Simons Institute for the Theory of Computing. Email: \email{qipengliu0@gmail.com}}
  \and
  Luowen Qian\footnote{Boston University. Email: \email{luowenq@bu.edu}}
}
\fi
\date{}
\fi

\begin{document}

\maketitle

\ifshowabs
\begin{abstract}
  Chandran et al.\,(SIAM J. Comput.~'14) formally introduced the cryptographic task of position verification, where they also showed that it cannot be achieved by classical protocols.
  In this work, we initiate the study of position verification protocols with \emph{classical} verifiers.
  We identify that proofs of quantumness (and thus computational assumptions) are necessary for such position verification protocols.
  For the other direction, we adapt the proof of quantumness protocol by Brakerski et al.\,(FOCS~'18) to instantiate such a position verification protocol.
  As a result, we achieve classically verifiable position verification assuming the quantum hardness of Learning with Errors.
  
  Along the way, we develop the notion of 1-of-2 non-local soundness for a natural non-local game for 1-of-2 puzzles, first introduced by Radian and Sattath (AFT~'19), which can be viewed as a computational unclonability property.
  We show that \oneoftwo non-local soundness follows from the standard 2-of-2 soundness (and therefore the adaptive hardcore bit property), which could be of independent interest.
\end{abstract}
\fi

\section{Introduction}
Position verification is the central task for position-based cryptography~\cite{CGMO09PBC}, which aims to verify one's geographical location in a cryptographically secure way.
The main technique is distance bounding, which infers the location assuming no faster-than-light communications from special relativity by placing timing constraints on the protocol.

The work of Chandran et al.\,\cite{CGMO09PBC} first formalized the task of position verification.
They in addition showed that it is impossible to achieve via any classical protocol where all the parties are classical.
Specifically, a few colluding adversaries can always efficiently convince the verifiers of an incorrect position, even with the help of computational assumptions.
As a result, all known classical position verification protocols that are secure against multiple adversaries, make hardware assumptions on the adversaries~\cite{CGMO09PBC,brodykalai-tcc2017-pbc}.

However, it turns out the attack above does not extend when the parties exchange quantum information.
The attack requires the adversaries to store the messages from the verifiers and at the same time forward them to the other adversaries, which violates the no-cloning theorem when the messages are quantum states unknown to adversaries.
A long line of work~\cite{beigi2011simplified, KMS11qubit-routing,Tomamichel_2013, unruh2014qrom, buhrman2014position, bluhm2021positionbased,DS2021} explored this idea by constructing protocols with BB84 states (or other similar states \cite{BFSS13gardenhose,allerstorfer2021new,junge2021banachspaces}), and proving them to be unconditionally secure.
Intuitively, these protocols get around the impossibility as these BB84 states are information theoretically unclonable when the adversaries receive them.

\paragraph*{Downsides of Quantum Communications.}
There are a lot of drawbacks for using quantum communication, especially under the context of position verification.

First and foremost, transmitting quantum information with fault tolerance is much more challenging.
As position verification is only meaningful with free-space (wireless) transmission, any practical protocol must be subject to a high loss.
In fact, Qi and Siopsis~\cite{qi2015loss} have shown that many known protocols stop working (lose either completeness/correctness, or soundness/security) when the error rate is above some threshold.
Unlike quantum key distribution, the parties in position verification do not share an authenticated classical channel, and must follow strict timing constraints, so techniques there do not generically carry over.
Furthermore, prior to our work, there was no known construction of fully loss tolerant position verification protocols against entangled adversaries, meaning being tolerant to any loss bounded away from 1.

Another issue arises when we consider high dimensions (2D or higher), which is that the parties must also send the quantum messages in the desired direction with high accuracy, or they would incur an even higher loss in transmission.
In practice, this is usually mitigated via a tracking laser \cite{ursin2006free,schmitt2007experimental}, %
although not perfectly.
If the BB84 state is naively broadcasted, the adversaries could obtain one copy each and therefore completely break the protocol.

Finally, adding other properties to the protocol is more difficult and inefficient when the communication is quantum.
For example, one could desire to authenticate the messages sent by the verifiers in order to protect the prover from revealing his location to other untrusted verifiers.
Unfortunately, authenticating a quantum message is highly nontrivial~\cite{BarnumCGST02,alagic2018can}.

All of these issues can be trivially resolved if the communication is classical.

One approach to remove quantum communication is to have the verifiers and the prover pre-share entanglement and use teleportation to transmit quantum messages over a classical channel.
However, this generic approach consumes the entanglement and therefore is undesirable if they would like to run the protocol multiple times for a considerable time.
Furthermore, it would require the parties to keep the entanglement coherent before the protocol begins, which can be expensive.

\ifexabs
  \let\subsection\section
\fi

\subsection{Our Results}

In this work, we show how to construct position verification protocols with classical verifiers, showing that quantum communication is not necessary for position verification without hardware assumptions.
Our main result is the following.

\ifexabs\begin{theorem}\else\begin{theorem}[Restatement of \Cref{cor:cvpv}]\fi
  \label{thm:main}
  Assuming the quantum (polynomial) hardness of Learning with Errors (LWE), there exists a classically verifiable position verification (CVPV) protocol with almost perfect completeness and negligible soundness against polynomial-time adversaries without pre-shared entanglement.
\end{theorem}

Our construction of the CVPV protocol is inspired by the (classically verifiable) proof of quantumness protocol by Brakerski et al.\,\cite{brakerski2018cryptographic}, which is proven secure under the same LWE assumption.

We also proved two variations of the theorem to handle adversaries with entanglement, albeit either assuming a stronger assumption or proven in an ideal model.

\ifexabs\begin{theorem}\else\begin{theorem}[Restatement of \Cref{thm:bounded-entanglement}]\fi
  \label{thm:main-subexp}
  Assuming the quantum \varul{subexponential} hardness of LWE, there exists a CVPV protocol with almost perfect completeness and inverse-subexponential soundness against \varul{bounded-entanglement} subexponential-time adversaries.
\end{theorem}

\ifexabs\begin{theorem}\else\begin{theorem}[Restatement of \Cref{cor:cvpv-qrom}]\fi
  \label{thm:main-qrom}
  Assuming the quantum hardness of LWE, there exists a CVPV protocol with almost perfect completeness and negligible soundness against \varul{unbounded-entanglement} polynomial-time adversaries \varul{in the quantum random oracle model}.
\end{theorem}

The quantum random oracle model (QROM), introduced by Boneh et al.\,\cite{AC:BDFLSZ11}, captures generic quantum attacks against cryptographic hash functions, modeled by random functions.

To the best of our knowledge, our protocols matches the state of the art in quantum position verification in terms of the entanglement bound.
All previous protocols in the standard model (as opposed to the QROM) are not known to be secure even against an arbitrary polynomial amount of entanglement, and any protocol can be broken with an exponential amount of entanglement \cite{buhrman2014position,beigi2011simplified}.
Furthermore, the only position verification protocol that is secure against any polynomial amount of entanglement that we are aware of is also proven in the QROM \cite{unruh2014qrom}.

\ifexabs We \else In \Cref{sec:attack}, we \fi
further show that there are also efficient attacks against \Cref{thm:main,thm:main-subexp} if the adversaries are allowed to pre-share more entanglement than what the entanglement bound allows.

Finally, for the other direction, we show that our assumption is somewhat minimal.
The classical impossibility easily extends if the prover is classical.
On a high level, if the adversaries can run in exponential time, the prover can always be simulated classically as her inputs and outputs are all classical; therefore, we would run into the classical impossibility.

Formally, we strengthen this intuition \ifexabs\else to \Cref{thm:necessity-poq} \fi to show that proofs of quantumness are necessary for any construction of classically verifiable position verification, even if we relax the requirement for position verification to be sound only against classical adversaries.
Since the prover response in a proof of quantumness could be simulated by a $\mathsf{PostBQP} = \mathsf{PP}$ machine\footnote{The idea is that one can capture simulation of the quantum prover as a sampling variant of the $\mathsf{PostBQP}$ problem as follows: simulate the quantum prover's next classical message given the current classical transcript of the protocol.}~\cite{postbqp}, as a consequence, it is impossible to construct unconditionally-sound proofs of quantumness (and thus classically verifiable position verification) without proving $\mathsf{PP} \not\subseteq \mathsf{BPP}$, even if we only consider position verification protocols with classical communications and quantum verifiers.

\subsection{Technical Overview}

\ifshowabs
  \paragraph*{Quantum Position Verification with One Quantum Message.}
We first recall the position verification protocol investigated by many works \cite{beigi2011simplified, KMS11qubit-routing,Tomamichel_2013, buhrman2014position,bluhm2021positionbased}.
The protocol has the property that only one message is quantum, and the only quantum requirement on the verifiers is to generate BB84 states.

Consider that in one-dimensional spacetime, there are two verifiers $V_0, V_1$, wishing to verify that the prover $P$ is located at a specific position somewhere between them.
At the beginning of the protocol, $V_0$ sends a BB84 qubit $H^\theta\ket x$ (where $\theta, x$ are uniformly random bits), and $V_1$ sends a classical bit $\theta$, so that they arrive at the prover's claimed position at the same time.
$P$ is supposed to measure the qubit in basis $\theta$ and return the measurement result to both verifiers.
At the end, the verifiers check that the prover's measurement result is $x$, and that they have received the responses ``in time''.

The intuition of the security proof is the following.
Consider an adversary $A_0$ located in between $V_0$ and $P$, and another adversary $A_1$ in between $P$ and $V_1$.
When $A_0$ receives the qubit, he does not yet know the basis $\theta$, and therefore he cannot immediately measure it.
However, if they decide to wait until $\theta$ is received, then either $A_0$ or $A_1$ will not have enough time to know the measurement result and send it to the verifiers.
Therefore, it seems if they want to answer correctly in time on both ends, $A_0$ must somehow produce two copies of the BB84 state, which is impossible as the state is information-theoretically unclonable without knowing the basis $\theta$.

\paragraph*{Computationally Unclonable States from Trapdoor Claw-free Functions (TCFs).}
As we have discussed, CVPVs require a proof of quantumness.
Therefore, a natural starting point is to open up the construction of the LWE proof of quantumness protocol by Brakerski et al.\,\cite{brakerski2018cryptographic}, and look for a similar unclonability property.

The proof of quantumness protocol could be described under the \oneoftwo puzzle framework by Radian and Sattath~\cite{radian2020semi}.
In particular, both trapdoor claw-free functions (TCFs) and noisy trapdoor claw-free functions (NTCFs) can be used to instantiate \oneoftwo puzzles.
However, we only have constructions of NTCFs from quantum LWE.
For this overview, we will work with the more intuitive notion of TCFs and use the \oneoftwo puzzle framework in the main technical body.

A TCF family is a family of efficiently computable $2$-to-$1$ functions $f_\pk: \{0, 1\}^n \to \mathcal Y$.
``Trapdoor'' means that with the trapdoor $\td$, one can efficiently invert the corresponding $f_\pk$ and get the two pre-images $x_0, x_1$.
``Claw-free'' means that without the trapdoor, it is hard for any polynomial-time quantum algorithms to find a collision for a random $f_\pk$.

The proof of quantumness protocol works as follows.
The verifier starts by sampling $\pk$ along with the trapdoor $\td$, and sends $\pk$ to the prover.
The prover prepares a uniform superposition over $\{0, 1\}^n$, computes $f_\pk$ on the superposition coherently, measures the image register to obtain $y \in \mathcal Y$, and sends $y$ as his response.
As $f_\pk$ is $2$-to-$1$, the residual state of the prover is
\begin{equation}
  \label{eq:introclawsuperposition}
  \frac{1}{\sqrt{2}}(\ket{x_0} + \ket{x_1}),
\end{equation}
where $x_0, x_1$ are the two pre-images of $y$.
The protocol concludes with the verifier sending a uniformly random challenge $b$ to the prover, and the prover measuring \eqref{eq:introclawsuperposition} either in the standard basis or the Hadamard basis.

If the prover is asked to measure in the standard basis, the measurement outcome will be a uniformly random $x$ which is either $x_0$ or $x_1$.
If the prover is asked to measure in the Hadamard basis, the measurement outcome will be a uniformly random $d$ such that $d \cdot (x_0 \oplus x_1) = 0$ over $\mathbb F_2^n$.
Since the verifier has the trapdoor, he can obtain $x_0, x_1$ by inverting $y$, and thus check whether the measurement outcome satisfies the requirements above.

As for security, we need an additional property called \emph{adaptive hardcore bit}, which says that any efficient quantum algorithm given $\pk$, cannot produce $y, x, d$ that passes the two checks simultaneously with probability significantly higher than \textonehalf, i.e. $f_\pk(x) = y$, $d \cdot (x_0 \oplus x_1) = 0$, and $d \ne 0$.
To see that this implies the proof of quantumness property, assume a classical prover can pass this proof of quantumness protocol with probability $1$, then we can always extract both $x$ and $d$ with probability $1$ by simply rewinding the classical prover.

In fact, the adaptive hardcore bit property also implies that the state \eqref{eq:introclawsuperposition} must be computationally unclonable.
This is simply because if somehow we can prepare two copies of this state, then measuring two copies in two bases will yield both $x$ and $d$.
This computational unclonability property has also been observed and used in prior works, in particular in the context of semi-quantum money \cite{radian2020semi} and two-tier quantum lightning \cite{kitagawa2020secure}.
Later we will see that the security proof for our CVPV protocol requires a stronger variant of computational unclonability than the ones considered in these works.

\paragraph*{Constructing CVPV.}
Given the setup, a natural idea for achieving CVPV is that instead of sending an unclonable state prepared by $V_0$, perhaps we can ask the prover (and hopefully also the adversaries) to prepare a quantum state that she herself cannot clone, similar to that in the proof of quantumness protocol.
Specifically, consider the CVPV protocol, where $V_0$ sends $\pk$ and $V_1$ sends $b$ with the same timing as before. In the end, they check that whether they have received the same prover response in time and whether the prover's measurement outcome passes the proof of quantumness check.
On the other hand, the prover in CVPV will run the prover in the proof of quantumness protocol and output $y, \ans$, where $y$ is the measured image of the superposition evaluation, and $\ans$ is the measurement outcome in the basis specified by $b$.

We now show that this construction already seems to get around the classical impossibility.
The attack from the impossibility is following: $A_0, A_1$ forwards the classical messages $\pk, b$ to each other, and at the end, they run the honest prover and send the output.
However, in this protocol, since the measurement performed by the prover has some nontrivial min-entropy, the verifiers will get two different responses with constant probability!
It is also not clear whether this computation could be simulated (almost) deterministically with shared randomness.
Certainly, if it could be simulated classically, then it would be breaking the proof of quantumness property.

Unfortunately, it turns out that a different attack completely breaks this CVPV protocol.
When $A_0$ receives $\pk$, he can simply runs the honest prover twice --- once on $b = 0$ and once on $b = 1$.
He obtains $y_0, \ans_0$ for $b = 0$ and $y_1, \ans_1$ for $b = 1$, and sends both of them to $A_1$.
On the other hand, $A_1$ simply forwards $b$.
Later, when both of them receive the message from each other, they pick $y_b, \ans_b$ as their responses to the verifiers.
It is not hard to show that this strategy simulates the prover perfectly.

We observe that in order for this attack to work, it is crucial that the adversaries can pick $y$ \emph{after} seeing $b$, which is impossible in the proof of quantumness protocol.
Therefore, to prevent this attack, our idea is to ``nudge'' the prover to the left, so that she can commit to $y$ \emph{before} seeing $b$.
More formally, the protocol is the same as before but the timing constraints are changed.
In particular, the verifiers make sure that the message $\pk$ reaches the prover a bit earlier than $b$, and at the end, they check that she should output $y$ as soon as she receives $\pk$ (and before she receives $b$).
\ifitcs\else
We refer the readers to \Cref{fig:prpv-protocol} for an illustration of the timing.
\fi

\paragraph*{Proving Soundness of CVPV.}
It turns out that with this simple fix, this CVPV can be proven secure.
\ifitcs
\else
In \Cref{claim:adversarial-behavior}, we show that according
\fi
to the timing constraints, we can again, without loss of generality, assume that there are two adversaries $A_0, A_1$, and that $A_0$ upon receiving $\pk$ needs to output $y$ to the verifiers \emph{immediately}, and after they receive a private communication from each other, they are supposed to produce two $\ans$'s to pass the verification.

We first consider a restricted set of adversarial strategies, called \emph{challenge-forwarding} adversaries, where the only restriction is that $A_1$ upon receiving $b$ simply forwards $b$ and does nothing else.
We claim that the success probability for challenge-forwarding adversaries cannot be significantly higher than \textthreequarters.

We now show that this suffices to show that the success probability for \emph{any} adversarial strategy without pre-shared entanglement cannot be significantly higher than \textthreequarters.
The proof is that assume $(A_0, A_1)$ breaks the CVPV with probability noticeably higher than \textthreequarters, we construct a challenge-forwarding adversary $(B_0, B_1)$ with the same success probability, which leads to a contradiction.
The construction of the reduction is similar to the attack for the first CVPV construction.
$B_0$, upon receiving $\pk$, runs $A_0$ on $\pk$ (and commits $y$) and simultaneously $A_1$ twice --- once on $b = 0$ and once on $b = 1$ --- and sends the residual state to the other party.
We can run $A_1$ twice as they do not pre-share entanglement.
Later, when both of them learn $b$, they can pick the correct execution to finish simulating $(A_0, A_1)$. 

\paragraph*{A (Computational) Non-Local Game for TCFs.}
What is left to be shown is that even challenge-forwarding adversaries cannot break the CVPV protocol.
\ifitcs
\else
In \Cref{thm:prpv-sound-forwarding}, we show
\fi
that for our protocol, what the adversaries can do is more or less equivalent to the following \emph{computational} (two-player) non-local game:
\begin{itemize}
  \item The game begins by announcing a TCF public key $\pk$.
  \item Two (computationally bounded) players $B$ and $C$ upon receiving $\pk$, agree on a classical ``commitment'' $y$. They then prepare a possibly entangled bipartite state $\rho_{BC}$ between themselves, after which they are separated.
  \item A \emph{single} challenge $b$ is then sampled uniformly at random and announced to $B$ and $C$ separately.
  \item $B$ and $C$ produce two answers $\ans_B$ and $\ans_C$ using $\rho_B$ or $\rho_C$ separately, and win the non-local game if both answers pass the proof of quantumness check with respect to $\pk, y, b$.
\end{itemize}

Another way to view this game is that it is the same as the TCF proof of quantumness protocol, except that after halfway, we ask the prover to run two copies of himself, i.e. split himself into two executions and finish each execution separately with the same verifier randomness.
If the prover's internal state was clonable, then the best prover's success probability should never decrease after the transformation.
Therefore, this can also be viewed as a computational unclonability property.

To prove the non-local soundness, assume that a strategy wins this non-local game significantly higher than \textthreequarters. We construct an algorithm breaking the adaptive hardcore bit property, by asking $B$ challenge $0$ (produce $x$) and $C$ challenge $1$ (output $d$).
On a high level, this reduction works because in a non-local game, the measurements made by $B$ and $C$ are on disjoint registers, and thus must be compatible no matter which challenges are given to them.

We now provide an informal proof that this reduction works for any non-signaling players.
A strategy is non-signaling if the marginal distribution for one player is independent of what the other player does, and the no signaling principle says that any bipartite measurement of a quantum state is non-signaling.
Let $W_0, W_1$ be the events where $B$ or $C$ produces a correct answer respectively in the non-local game.
We can rewrite the success probability of the non-local game to be $p := \Pr[W_0 \land W_1]$.
Then
$$p = \frac12 \Pr[W_0 \land W_1 | b = 0] + \frac12 \Pr[W_0 \land W_1 | b = 1] \le \frac12 \Pr[W_0 | b = 0] + \frac12 \Pr[W_1 | b = 1].$$
On the other hand, let $W'_0, W'_1$ be the events where $B$ or $C$ produces a correct answer respectively in the reduction, where $B$ receives challenge $0$ and $C$ receives challenge $1$.
Then the success probability of the reduction is $p' := \Pr[W'_0 \land W'_1]$.
$p' \le \frac12 + \negl$ since the reduction is efficient, and by union bound,
$$p' = 1 - \Pr[\lnot W'_0 \lor \lnot W'_1] \ge 1 - \Pr[\lnot W'_0] - \Pr[\lnot W'_1] = \Pr[W'_0] + \Pr[W'_1] - 1.$$
Notice that $\Pr[W'_0] = \Pr[W_0 | b = 0]$ by construction and the no signaling principle, and similarly $\Pr[W'_1] = \Pr[W_1 | b = 1]$.
The conclusion $p \le \frac34 + \negl$ follows by rearranging the terms.

The computational unclonability requirements in prior works \cite{radian2020semi,kitagawa2020secure} cannot be cast as a non-local game, since there the two players need to answer different challenges instead of the same one. Therefore, by adaptive hardcore bit property, the game is hard even if the two players can communicate.
We think that this computational non-local hardness that we achieve could potentially have applications to other quantum cryptography relying on the no-cloning principle.

\paragraph*{Soundness Amplification via Parallel Repetition.}
So far, we have shown how to construct a CVPV with soundness \textthreequarters\ against adversaries without pre-shared entanglement.

To achieve negligible soundness, one natural attempt is to do sequential repetition.
However, sequential repetitions are undesirable in our setting as (1) sequential repetitions will undesirably increase the number of rounds/time/complexity of the final protocol; (2) more crucially, adversaries can take advantages of a multiple round protocol and use quantum communication to share some entanglement even if they have no pre-shared entanglement at the beginning of the protocol.
\ifitcs
\else
Combining with the attack that we give in \Cref{sec:attack},
\fi
one can show that with sequential repetitions, the soundness does not decrease at all!

Therefore, we turn to consider parallel repetitions, which traditionally have been more technically challenging than sequential repetitions under numerous different contexts.
One difficulty is that our CVPV protocol can be viewed as a four-message private-coin interactive argument with additional structures, and therefore known transformations for interactive arguments do not apply.
Another difficulty is that a common technique for proving parallel repetition for private-coin arguments is to perform rejection sampling, which in our case of proving parallel repetition of CVPV, would lead to either communication or pre-shared entanglement between the adversaries, neither of which is allowed for this setting.

The key idea is that instead of proving a parallel repetition theorem for the CVPV protocol, we first establish a parallel repetition theorem for the TCF non-local game, where at least the two players are allowed to share entanglement.
We then construct a CVPV protocol with a stronger variant of the non-local game.
However, we still need to be careful about the reduction since in the non-local game, two players cannot communicate after $y$ is sent.

We first consider the parallel repetition where the non-local game is repeated $k$ times in parallel, except that we use a single challenge $b$ for all the executions.
We show that the non-local soundness can be decreased to \textonehalf\ if $k$ is large enough using known results \cite{radian2020semi} (which in turn uses a classical parallel repetition theorem \cite{canetti2005hardness}). The \textonehalf\ soundness here is tight as the adversaries can always guess $b$ correctly with probability \textonehalf.

We next consider a second parallel repetition where the strengthened game from above is repeated $k'$ times in parallel, and this time we use fresh random challenges for all the executions.
As the strengthened game has soundness \textonehalf, this implies that the two quantum predicates (standard basis test and Hadamard basis test) satisfy computational orthogonality, similar to the one that has appeared under a different application of parallel repetitions for TCFs, which is quantum delegation \cite{AlagicCGH20,chia2019classical}.
Therefore, using the ideas from those works, we show that the non-local soundness decreases exponentially in $k'$.

Finally, using the same reduction from non-local games to CVPV as before, we show that we can achieve the CVPV protocol with negligible soundness.

\paragraph*{Handling Entangled Adversaries.}
We have proven that our protocol is negligibly sound against adversaries without pre-shared entanglement.
It turns out that our protocol is similar enough to the previous 
quantum position verification that a lot of techniques there can be naturally ported here as well.

Using a standard trick \cite{aaronson2004limitations,Tomamichel_2013}, we can show that the protocol can be made secure against any adversaries with an a-priori-chosen polynomial amount of pre-shared entanglement, albeit requiring subexponential hardness of quantum LWE, as the reduction for parallel repetition needs to run in subexponential time.

On the other hand, our protocol can also be attacked with $n$ EPR pairs where $n$ is the length of the output of $f_\pk$.
The attack is very similar to the attack for the quantum position verification protocol we give in the beginning.
The adversaries simply prepare the state \eqref{eq:introclawsuperposition} honestly (which we recall is the only non-timing-wise change to the protocol) and perform the attack against the base protocol.
In particular, they teleport the state using EPR pairs to perform measurements in a homomorphic way, whose outcome later they can recover with one round of communication.
Attacking the protocol after parallel repetition can be done by running the attack above in parallel.

Finally, we modify the CVPV protocol into the QROM to prove that it is sound against unbounded entanglement, where the modification is very similar to how Unruh~\cite{unruh2014qrom} modifies the base position verification protocol into the QROM.
On a high level, the attack for the previous protocol works because the honest prover's operation after committing $y$ is a Clifford.
With Unruh's transformation, the operation now involves evaluating a random function, which %
cannot be efficiently computed by a Clifford circuit.
The security proof in the QROM from Unruh's work also carries over, except here we reduce the adversarial strategy with entanglement against the QROM CVPV, to the TCF non-local game after parallel repetition (in the standard model), instead of a monogamy-of-entanglement game~\cite{Tomamichel_2013}.

\subsection{Future Directions}

\paragraph*{High Dimensional Position Verification.}
We conjecture that the following construction, inspired by the position verification protocol of Unruh~\cite{unruh2014qrom}, could be secure in higher dimensions under the quantum random oracle model (QROM) using the ideas from Unruh:
\begin{enumerate}
  \item $V_0$ broadcasts $\pk$.
  \item $V_0, ..., V_n$ sample uniformly random strings $x_0, ..., x_n$ respectively and broadcast them.
  The timing is done so that these strings arrives  at the prover a bit later than $\pk$.
  \item At the end, the $(n + 1)$ verifiers check that the prover answers arrive in time, and passes the check with respect to challenge $H(x_0 \oplus \cdots \oplus x_n)$, where $H$ is the random oracle.
\end{enumerate}

\paragraph*{Time-Entanglement Trade-Offs: Upper and Lower Bounds.}
Classically verifiable position verification protocols have the curious feature of being completely broken against classical adversaries with unbounded computational power, as they can simulate the honest quantum execution.
On the other hand, our protocol can be efficiently broken using a linear amount of entanglement but secure against adversaries with bounded entanglement.
This suggests that there may be some time-entanglement trade-offs for the optimal attack.
Clearly, the trivial trade-off to attack the CVPV after parallel repetition is that the adversaries can use their entanglement to break some copies, and brute-force the rest of the copies.
It is interesting whether there is a significantly better time-entanglement trade-offs that could be achieved for attacking this protocol or classically verifiable position verification protocols in general.

For the other direction, we also wonder if there is a tighter lower bound on the entanglement than what we prove.

\paragraph*{Decreasing Quantum Memory for the Prover.}
We have shown in
\ifitcs\else\Cref{thm:bounded-entanglement}\fi\ 
that assuming subexponential hardness of quantum LWE, we can construct classically verifiable position verification protocols that is secure against any a-priori bounded entanglement.
Unfortunately, in our protocols, even the honest prover needs to keep his quantum memory (which is of length $\tilde O(\lambda)$ when entanglement bound is 0) coherent for some time, and the size of the quantum memory is even larger than the entanglement bound.
Indeed, we have also shown that if the adversaries share as much entanglement as the size of the honest prover's quantum memory, then the protocol can be efficiently broken.
However, the adversaries might need to keep the entanglement coherent long before the protocol begins, and this might be much longer than the duration needed by the honest prover.

Nevertheless, it would be interesting if we can avoid this drawback.
We therefore ask whether it is possible to come up with provably secure CVPV protocols where the honest prover's quantum memory is smaller than the entanglement bound in the standard model, or maybe even without any quantum memory at all.

\paragraph*{Weakening the Assumption.}
We show how to achieve CVPV assuming quantum hardness of LWE, which is a cryptographic assumption.
Can we relax this assumption further?
One possible assumption is the existence of a classically verifiable quantum sampling task satisfying some requirements.

\else
In this section, we summarize the ideas used for the main theorems and omit many details in interest for space.
We refer the interested readers to the full version for the missing details.

On a high level, our protocol can be thought of as asking the prover to prepare a computationally unclonable state by herself --- instead of sending her an information theoretically unclonable BB84 state from a verifier --- and having the prover measure the state.
Such a computationally unclonable state can be constructed via adapting the proof of quantumness protocol by Brakerski et al.\,\cite{brakerski2018cryptographic}
However, we need to carefully design the timing constraints of the protocol in an unconventional way in order to force the adversaries into preparing this unclonable state as well.

Towards establishing the security, we show that the soundness of our construction can be reduced to the hardness of a natural non-local game for noisy trapdoor claw-free functions (NTCFs).
We then prove that such a non-local game is indeed hard assuming the adaptive hardcore bit property, in fact against any non-signaling players, which includes non-communicating quantum players.

In the end, we get a classically verifiable position verification with constant soundness.
In order to reduce the soundness further, we consider some parallel repetition of the protocol, and prove that it is negligibly sound.
In particular, we suitably adapt the techniques used for analyzing the parallel repetition of NTCFs under other contexts like semi-quantum money \cite{radian2020semi} and quantum delegation \cite{AlagicCGH20,chia2019classical}.

Finally, since our protocol is somewhat similar to the quantum position verification protocols with BB84 \cite{beigi2011simplified, KMS11qubit-routing,Tomamichel_2013, buhrman2014position,bluhm2021positionbased}, known techniques there can be suitably ported to our construction to achieve additional properties we need for \Cref{thm:main-subexp,thm:main-qrom}.
\fi

\ifexabs\else
\ifblind\else
\section*{Acknowledgements}
\acks
The authors would also like to thank the anonymous reviewers from ITCS 2022 and QIP 2022 for their kind and thoughtful comments.

Jiahui Liu is \jiahuifunding
Qipeng Liu is \qipengfunding
Luowen Qian is \luowenfunding
\fi

\section{Notations}

We refer the readers to \cite{nielsen2002quantum} on basic quantum information and computation concepts.
As our work only works with the Learning with Error (LWE) assumption indirectly, we refer the readers to \cite[Section 2.3]{brakerski2021cryptographic} for further information on the assumption.

We call a function $f: \mathbb N^+ \to \mathbb R^{\ge 0}$ negligible ($f(n) = \negl(n)$) if for any $g(n) = \poly(n) := n^{O(1)}$, $f(n) \le 1/g(n)$ for all sufficiently large $n$.
Throughout this paper, we use $\lambda$ to denote the security parameter unless specified otherwise.

Let $\Hs$ denote a finite-dimensional Hilbert space.
We use Dirac notation to express vectors which represent pure states, for example $\ket{\psi}$.
We let $\mathcal S(\Hs)$ denote the set of all density operators, which are positive semidefinite operators on $\Hs$ with trace 1, and represent mixed states.

Quantum registers simply mean a collection of qubits in a given state.
Consider a mixed state $\rho_{AB}$ where the qubits are partitioned into sets $A$ and $B$.
We denote $\rho_A$ to refer to the qubits in $A$ in state $\rho_{AB}$.

The quantum random oracle model (QROM) \cite{AC:BDFLSZ11}, is the model where a single function $f: \mathcal X_\lambda \to \{0, 1\}^\lambda$ is sampled uniformly at random.
All parties get oracle access to the unitary $\mathcal O$ such that $\mathcal O \ket x \ket y = \ket x \ket{y \oplus f(x)}$ for all $x \in \mathcal X_\lambda, y \in \{0, 1\}^\lambda$.

\section{\oneoftwo Puzzles and Non-Local Soundness}

\subsection{1-of-2 Puzzles}

\begin{definition}[{\oneoftwo Puzzles\,\cite[Definition 2.1]{radian2020semi}}]
    A \oneoftwo puzzle $\Zs$ is a tuple of four efficient algorithms $(\keygen, \obligate, \solve, \ver)$, where:
    \begin{description}
        \item The key generation algorithm $\keygen$, is a classical algorithm that on security parameter $1^\lambda$, outputs a public key $\pk$ and a secret key $\sk$: $(\pk, \sk) \gets \keygen(1^\lambda)$. 
        \item The obligation algorithm $\obligate$, is a quantum algorithm that on input a public key $\pk$, outputs a classical string $y$ called the obligation and a quantum state $\rho$: $(y, \rho) \gets \obligate(\pk)$. 
        \item The \oneoftwo solver $\solve$, is a quantum algorithm that on input a public key $\pk$, an obligation $y$, a quantum state $\rho$ and a challenge bit $b$, outputs a classical answer $\ans$: $\ans \gets \solve(\pk, y, \rho, b)$. 
        \item The verification algorithm $\ver$, is a classical deterministic algorithm that on input a secret key $\sk$, an obligation $y$, a challenge bit $b$ and an answer $\ans$, it outputs $0$ or $1$: $\ver(\sk, y, b, \ans) \in \{0, 1\}$. 
    \end{description}
    Furthermore, it satisfies the following completeness and \twooftwo soundness.

\textbf{Completeness}\footnote{Our completeness slightly differs from the original definition in the sense that the negligible term is dropped from the definition, since unlike soundness, there is no additional quantifier on the adversary. This change is made also to signify the imperfect completeness.}:  Let $c$ be some function $c: \mathbb{N} \to \mathbb R$. We say that the \oneoftwo puzzle $\Zs$ has completeness $c$ if \begin{align*}
    \Pr_{b \gets \{0,1\}}\left[ \ver(\sk, y, b, \ans) = 1 : \begin{array}{cc}(\pk, \sk) \gets \keygen(1^\lambda) \\ (y, \rho) \gets \obligate(\pk) \\ \ans \gets \solve(\pk, y, \rho, b) \end{array}\right] \geq c(\lambda).
\end{align*}

\textbf{\twooftwo Soundness}\footnote{We use a slightly different notion of \twooftwo soundness instead of the \twooftwo hardness in the original work. In particular, the original definition of having \twooftwo hardness $(1 - h)$ is equivalent to having \twooftwo soundness $h$.}: Let $s: \mathbb{N} \to [0, 1]$ be a function. We say that the \oneoftwo puzzle $\Zs$ has \twooftwo soundness $s$ if for any QPT \twooftwo solver $\Ts$, there exists a negligible function $\negl(\lambda)$ such that
\begin{align*}
    \Pr\left[  {\sf SOLVE2}_{\Ts, \Zs}(1^\lambda) = 1 \right] \leq s(\lambda) + \negl(\lambda),
\end{align*}
where the \twooftwo solving game ${\sf SOLVE2}_{\Ts, \Zs}(1^\lambda)$ is defined as the following:
\begin{enumerate}
    \item The challenger runs $(\pk, \sk) \gets \keygen(1^\lambda)$. 
    \item The \twooftwo solver $\Ts$ receives public key $\pk$ and outputs a triple of classical messages $(y, \ans_0, \ans_1)$.
    \item The game outputs 1 if and only if (note that $\ver$ is a classical deterministic function)
        \begin{align*}
            \ver(\sk, y, 0, \ans_0) = 1 \,\wedge\, \ver(\sk, y, 1, \ans_1) = 1. 
        \end{align*}
\end{enumerate}
We say $\Zs$ is a $(c, s)$-\oneoftwo puzzle if it has completeness $c$ and \twooftwo soundness $s$. 
\end{definition}

In the same work \cite{radian2020semi}, they also show that NTCFs (\Cref{def:trapdoorclawfree}) implies a \oneoftwo puzzle.
We further elaborate the connection between NTCFs and \oneoftwo puzzles in \Cref{sec:NTCF_prelim}.

\begin{theorem}[{\cite[Theorem 2.2]{radian2020semi} and \Cref{thm:lwe-ntcf}}]
\label{thm:NTCFis1of2}
    An NTCF implies a $(1 - \negl, \frac{1}{2})$-\oneoftwo puzzle. Therefore, $(1 - \negl, \frac{1}{2})$-\oneoftwo puzzles exist assuming quantum hardness of LWE. 
\end{theorem}

\subsection{1-of-2 Puzzle as a Non-Local Game}
\label{sec:one_two_puzzle_nonlocal}

We now define a non-local game for \oneoftwo puzzles, and show the connection between the success probability of the non-local game and the \twooftwo soundness of the underlying \oneoftwo puzzle.

\begin{definition}[Non-Local Games of \oneoftwo Puzzles] \label{def:non-local_game}
    Let $\Zs$ be a $(c, s)$-\oneoftwo puzzle. The non-local solving game $\text{\sf NON-LOCAL-SOLVE}_{\Ws, \Zs}(1^\lambda)$ for any non-local player $\Ws = (\As, \Bs, \Cs)$, where $\As, \Bs, \Cs$ are three quantum algorithms, is defined as follows:
    \begin{itemize}
        \item The challenger runs $(\pk, \sk) \gets \keygen(1^\lambda)$. 
        \item The algorithm $\As$ receives public key $\pk$ and outputs a classical obligation $y$ together with a quantum state $\sigma_{BC} \in \mathcal S(\mathcal{H}_B \otimes \mathcal{H}_C)$: $(y, \sigma_{BC}) \gets \As(\pk)$.
        It commits $y$ to the challenger. 
        
        $\Bs$ receives $\sigma_B$ and $\Cs$ receives $\sigma_C$. 
        \item The challenger samples a challenge bit $b \gets \{0, 1\}$ and sends $b$ to both $\Bs$ and $\Cs$. 
        \item $\Bs$ and $\Cs$ perform some local computations and then output $\ans_{\Bs} \gets \Bs(\sigma_B, b)$ and $\ans_{\Cs} \gets \Cs(\sigma_C, b)$ respectively. 
        \item The game outputs $1$ if and only if both $\Bs, \Cs$ answer correctly, i.e.
        \begin{align*}
            \ver(\sk, y, b, \ans_{\Bs}) = 1 \,\wedge\, \ver(\sk, y, b, \ans_{\Cs}) = 1. 
        \end{align*}
    \end{itemize}
\end{definition}

\begin{definition}[\oneoftwo Non-Local Soundness] \label{def:non-local-soundness}
Let $\tau: \mathbb{N} \to \mathbb R$ be an arbitrary function. We say that the \oneoftwo puzzle $\Zs$ has \oneoftwo non-local soundness $\tau$ if for any QPT non-local player $\Ws$, there exists a negligible function $\negl(\lambda)$ such that
\begin{align*}
    \Pr\left[  \text{\sf NON-LOCAL-SOLVE}_{\Ws, \Zs}(1^\lambda) = 1 \right] \leq \tau(\lambda) + \negl(\lambda).
\end{align*}
\end{definition}

We now establish the connection between \oneoftwo non-local soundness and \twooftwo soundness.
\begin{theorem}\label{thm:2of2vsnon-local}
    Let $\Zs$ be a \oneoftwo puzzle with \twooftwo soundness $s$.
    $\Zs$ has \oneoftwo non-local soundness $\tau = (s+1)/2$. 
\end{theorem}
\begin{proof}%
    Let $\Ws = (\As, \Bs, \Cs)$ be any QPT non-local player for $\textsf{NON-LOCAL-SOLVE}_{\Ws, \Zs}(1^\lambda)$ that achieves success probability $\tau = \tau(\lambda)$.
    Let $(\pk, \sk) \gets \keygen(1^\lambda)$ and $(y, \sigma_{BC}) \gets \As(\pk)$.
    Let $\PB_{\sk, y, b}, \PC_{\sk, y, b}$ be the projection acting on register $B, C$ respectively corresponding to the predicate $\ver(\sk, y, b, \ans) = 1$ where $\ans$ is either the output of $\Bs$ or $\Cs$.
    Using this notation, we can rewrite the success probability:
    \begin{align*}
        \tau = \Pr\left[  \text{\sf NON-LOCAL-SOLVE}_{\Ws, \Zs}(1^\lambda) = 1 \right] %
        =   \mathop{\E}_{\substack{\keygen, \As, b \gets \{0, 1\} } } \Tr \left[  \left({\PB}_{\sk, y, b} \otimes {\PC}_{\sk, y, b} \right) \sigma_{BC} \right]. 
    \end{align*}
    Here the subscripts $\keygen$ and $\As$ stands for the randomness of sampling $(\pk, \sk) \gets \keygen(1^\lambda)$ and the measurement randomness from $(y, \sigma_{BC}) \gets \As(\pk)$. 
    
    Expanding the expectation on $b$ and using the linearity of expectation and the trace operator, we get
    \begin{align*}
     2\tau = \mathop{\E}_{\keygen, \As}\Tr\left[  \left({\PB}_{\sk, y, 0} \otimes {\PC}_{\sk, y, 0} +  {\PB}_{\sk, y, 1} \otimes {\PC}_{\sk, y, 1} \right) \sigma_{BC} \right].
    \end{align*}
    Since $\PB_{\sk, y, b}, \PC_{\sk, y, b} \preccurlyeq \I$ for any $\sk, y, b$, we have:
    \begin{align*}
     2\tau &\le \mathop{\E}_{\keygen, \As}\Tr\left[  \left({\PB}_{\sk, y, 0} \otimes \I + \I \otimes {\PC}_{\sk, y, 1} \right) \sigma_{BC} \right]  \\ 
     &\le \mathop{\E}_{\keygen, \As}\Tr\left[\left(\I \otimes \I + {\PB}_{\sk, y, 0} \otimes {\PC}_{\sk, y, 1} \right) \sigma_{BC} \right]  \\
     &= 1 + \mathop{\E}_{\keygen, \As}\Tr\left[\left({\PB}_{\sk, y, 0} \otimes {\PC}_{\sk, y, 1} \right) \sigma_{BC} \right],
    \end{align*}
    where the second inequality is simply due to the fact that for any $P, Q \preccurlyeq \I$, $0 \preccurlyeq (\I - P) \otimes (\I - Q) = \I \otimes \I - P \otimes \I - \I \otimes Q + P \otimes Q$.
    Therefore,
\begin{align*}
     \E_{\keygen, \As}\Tr\left[\left(\PB_{\sk, y, 0} \otimes \PC_{\sk, y, 1} \right) \sigma_{BC} \right]  \geq  2\tau - 1.
\end{align*}

Now we construct a \oneoftwo puzzle solver $\Ts$, whose success probability is exactly the left hand side of the inequality above:
 \begin{enumerate}
     \item $\Ts$ upon receiving $\pk$, runs $\As(\pk)$ to produce $(y, \sigma_{BC})$. 
     \item It runs $\Bs$ and $\Cs$ on $\rho$ with different challenge bits $0, 1$ respectively, and outputs $\ans_0 = \ans_{\Bs}, \ans_1 = \ans_{\Cs}$. 
     \item It outputs $(y, \ans_0, \ans_1)$.
 \end{enumerate}
$\Ts$ solves \twooftwo puzzle if and only if both $\ans_0, \ans_1$ pass the verification, which in turn is at least $2\tau - 1$ as argued above.

On the other hand, since $\Ws$ is efficient, so is $\Ts$.
Therefore, the success probability of $\Ts$ is at most $s + \negl$.
Thus, $2\tau - 1 \le s + \negl$.
We conclude the proof by simply rearranging the terms.
\end{proof}

Combined with \Cref{thm:NTCFis1of2}, we get the following corollary: 
\begin{corollary}
\label{cor:1of2nonlocal}
    Assuming the quantum hardness of LWE, there exists a \oneoftwo puzzle that has completeness $1 - \negl$ and \oneoftwo non-local soundness \textthreequarters.
\end{corollary}

\section{Towards Position Verification}
In this section, we formally introduce the model and the definition of position verification.
Starting with a \oneoftwo puzzle with completeness $c$ and \oneoftwo non-local soundness $\tau$, we then show how to achieve a position-robust position verification with completeness $c$ and soundness $\tau$ against adversaries without entanglement.
Combining with the \oneoftwo puzzle from LWE, this gives a non-trivial position verification with almost perfect completeness and constant soundness.
We defer further decreasing soundness and handling entanglement to the next section.

\subsection{The Vanilla Model}

We restrict our attention to position verification in one dimension.
We consider the same model as the Vanilla Model (the standard model) from \cite{CGMO09PBC}, but augment it with quantum capabilities for the prover and the adversaries, as we only work with classically verifiable protocols.
We have three types of parties: prover, verifier, and adversary.
\begin{itemize}
    \item Space and time are continuous.
    \item The clocks of all parties are synchronized.
    \item Before the protocol begins, all parties are given as input the position of all verifiers, the claimed position of the prover, and a security parameter $\lambda$.
    \item The prover and adversaries have quantum computation capabilities, whereas the verifiers are entirely classical.
    \item The verifiers share a private trusted classical communication channel.
    \item All prover-verifier communications are classical broadcast messages. However, adversaries can send directional messages to any specific verifier that expects to receive a broadcast instead.
    \item Adversaries can also use a private (quantum) communication channel, so that the verifiers will not detect any malicious activity.
    \item All computations are done instantaneously, but all messages in all channels travel at speed 1 (the speed of light).
    
\end{itemize}

We first recall the usual completeness and soundness requirements of position verification protocols.
Since physical space satisfies translational symmetry, without loss of generality, we assume 
the claimed position of the prover is fixed a priori (say to be 1), instead of being an input to the parties.

\begin{definition}
  Let $c: \mathbb{N}^+ \to \mathbb R$.
  We call a position verification protocol to have \emph{completeness} $c$, if the prover, located at 1, can convince the verifiers with probability at least $c(\lambda)$ for any security parameter $\lambda$.
\end{definition}
\begin{definition}
  An \emph{adversarial strategy} is specified by a list of pairs $(p_i, A_i)$, where $p_i$ is the location of adversary $i$ and $A_i$ is the interactive (quantum) Turing machine it runs.
  A family of adversarial strategies is a list of adversarial strategies indexed by the security parameter $\lambda$.
\end{definition}

In this work, we focus on the setting where the families can be efficiently uniformly generated, i.e. there exists a deterministic polynomial-time Turing machine $M$ such that a useful description of the adversarial strategy for $\lambda$ can be efficiently generated by $M$ on input $1^\lambda$.
We abuse the notation to omit ``family'' whenever the context is clear.

\begin{definition}
  Let $s: \mathbb{N}^+ \to \mathbb R$ and $\mathcal S$ be any set of adversarial strategies.
  We call a position verification protocol to have \emph{soundness} $s$ against $\mathcal S$, if for any family of adversarial strategies $S \in \mathcal S$, using strategy $S$ can convince the verifiers with probability at most $s(\lambda) + \negl(\lambda)$ for some negligible function $\negl(\cdot)$.
  If $s = 0$, we call the protocol to have negligible soundness against $\mathcal S$.
\end{definition}

In the literature, we usually consider one round protocols where the prover receives two messages, performs one computation, and sends two responses.
In this case, it is easy to see that for soundness, instead of taking $\mathcal S$ to be the largest possible set of adversarial strategies, which is all strategies that could occupy the entire space outside of the claimed position, it is equivalent to only consider strategies with two adversaries, one on each side.

\subsection{Position Robustness}

We now introduce the position-robust version of these requirements.
Again since the spacetime we consider here is unitless, we without loss of generality assume the prover claims that it is somewhere in $(1, 2)$.

\begin{definition}
  Let $c: \mathbb{N}^+ \to \mathbb R$.
  We call a position verification protocol to have \emph{position-robust completeness} $c$, if the prover, located at \emph{anywhere} in $(1, 2)$, can convince the verifiers with probability at least $c(\lambda)$ for all $\lambda$.
\end{definition}

This notion is natural from the practical point of view.
Since the position measurement device always has some errors, a non-robust position verification protocol can never have any practical value.
Another reason we consider this notion is that we do not know how to make our protocol non-robust --- later we will see that in our construction, neither the prover nor the adversaries can occupy point 2.

Naturally, we should modify the set of adversarial strategies that we should consider for soundness under the position-robust setting.
\begin{itemize}
    \item The most general set of strategies, denoted by $\mathcal R$, is the set of strategies with the only restriction that all positions lie in $(-\infty, 1) \cup (2, +\infty)$.
    We allow the adversaries to do quantum setup before the protocol begins, including setting up entanglement between them.
    
      We do not consider adversaries at 1 or 2; that is, we do not allow the prover nor the adversaries to be at point 1 or 2.
      While it might be interesting to try to extend either completeness or soundness to close this small gap, these two single points have measure 0 and thus we consider it to be practically irrelevant --- indeed everything at the end will be subject to the precision of the devices being used, and the location gap caused by the device errors will greatly exceed these two points of failure.
    \item The set of polynomially bounded strategies (with pre-shared entanglement), denoted by $\mathcal R_P$, is the subset of $\mathcal R$ with further restrictions that there exists some polynomial $p(\cdot)$, such that the running time of the Turing machine that generates the strategy as well as that of every adversary is bounded by $p(\lambda)$.
      
      With polynomial hardness assumptions, $\mathcal R_P$ is the largest set of strategies that we will consider for the soundness of a \emph{classically verifiable} protocol, as an unbounded strategy from $\mathcal R$ can always convince the verifiers by simulating the quantum prover classically, and then using the attack strategy from \Cref{thm:necessity-poq}.
    \item The set of bounded-entanglement strategies, denoted by $\mathcal R_L$ for some function $L: \mathbb N^+ \to \mathbb N$, is the subset of $\mathcal R_P$ with further restrictions that the total \emph{quantum} communication between the adversaries \emph{before} the protocol begins, is bounded by $L(\lambda)$, which is an upper bound on the entanglement (measured via von Neumann entanglement entropy) that the adversaries share before the protocol begins.
\end{itemize}

Observe that for 1D, it suffices to consider at most two verifiers --- one on the left of the prover, the other on the right of the prover.
Indeed, if there are more than one verifiers on one side of the prover, the verifier that is closest to the prover can simulate all the interactions for the other verifiers.

Let the verifier on the left be $V_0$, and the one on the right be $V_1$.
We remark that for position-robust position verification, in general $V_0$ could be anywhere in $(-\infty, 1]$ and $V_1$ could be anywhere in $[2, +\infty)$.
However, since for all intents and purposes, sending a message to the right at $(x, t)$ is equivalent to sending it at $(x - \delta, t - \delta)$ (and vice versa), we can without loss of generality assume that $V_0$ and $V_1$ are at 0 and 3.

\subsection{The Base Protocol}

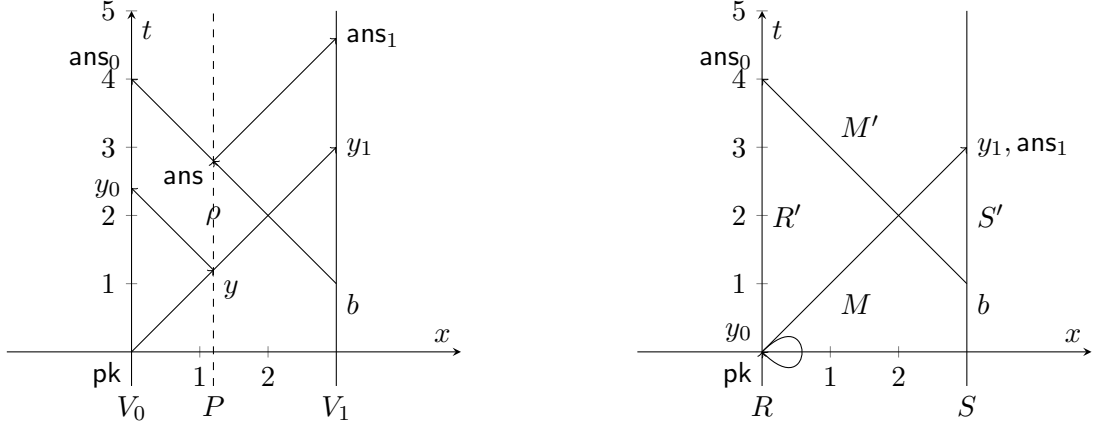
\begin{figure}
  \centering
  \begin{multicols}{2}
    \newcommand{\Px}{1.2}
\begin{tikzpicture}
\begin{axis}[
width=3in,
axis lines=middle,
axis equal,
xmin=0,xmax=3,
ymin=-0.5,ymax=5,
domain=0:3,
xlabel=$x$,ylabel=$t$,
xtick={0, 1, 2},
ytick={0, 1, 2, 3, 4, 5},
clip=false,
]
\draw [->]
  (0, 0) edge (\Px, \Px)
  (\Px, \Px) edge (0, 2 * \Px)
  (\Px, \Px) edge (3, 3)
  (3, 1) edge (\Px, 4 - \Px)
  (\Px, 4 - \Px) edge (0, 4)
  (\Px, 4 - \Px) edge (3, 7 - 2 * \Px)
;
\draw (3, -0.5) -- (3, 5);
\draw [dashed] (\Px, -0.5) -- (\Px, 5);
\node [below left] at (0, 0) {$\pk$};
\node [below right] at (\Px, \Px) {$y$};
\node [left] at (0, 2 * \Px) {$y_0$};
\node [right] at (3, 3) {$y_1$};
\node [below right] at (3, 1) {$b$};
\node [below left] at (\Px, 4 - \Px) {$\ans$};
\node [above left] at (0, 4) {$\ans_0$};
\node [right] at (3, 7 - 2 * \Px) {$\ans_1$};
\node at (\Px, 2) {$\rho$};
\node [below] at (0, -0.5) {$V_0$};
\node [below] at (\Px, -0.5) {$P$};
\node [below] at (3, -0.5) {$V_1$};
\end{axis}
\end{tikzpicture}
\let\Px\undefined
    \begin{tikzpicture}
\begin{axis}[
width=3in,
axis lines=middle,
axis equal,
xmin=0,xmax=3,
ymin=-0.5,ymax=5,
domain=0:3,
xlabel=$x$,ylabel=$t$,
xtick={0, 1, 2},
ytick={0, 1, 2, 3, 4, 5},
clip=false,
]
\draw [->]
  (0, 0) edge (3, 3)
  (3, 1) edge (0, 4)
  (0, 0) edge [out=45, in=-45, distance=1cm] (0, 0)
;
\draw (3, -0.5) -- (3, 5);
\node [below left] at (0, 0) {$\pk$};
\node [above left] at (0, 0) {$y_0$};
\node [right] at (3, 3) {$y_1, \ans_1$};
\node [below right] at (3, 1) {$b$};
\node [above left] at (0, 4) {$\ans_0$};
\node [below] at (0, -0.5) {$R$};
\node [right] at (0, 2) {$R'$};
\node [below] at (3, -0.5) {$S$};
\node [right] at (3, 2) {$S'$};
\node [below right] at (1, 1) {$M$};
\node [above right] at (1, 3) {$M'$};
\end{axis}
\end{tikzpicture}
  \end{multicols}
  \caption{
    Spacetime diagrams for $\prpv$ (left) and the possible adversarial behaviors shown in \Cref{claim:adversarial-behavior} (right).
    $\rho$ indicates the time interval for the prover to keep his quantum memory, and $R, R', S, S', M, M'$ indicate the quantum registers in the adversarial behavior.
    The small loop indicates that $A_0$ immediately outputs $y_0$ upon receiving $\pk$.
  }
  \label{fig:prpv-protocol}
\end{figure}

\begin{construction}[$\prpv$ Protocol]
\label{construction:prpv}
Let $\Zs = (\keygen, \obligate, \solve, \ver)$ be a \oneoftwo puzzle.
We define the $\prpv_\Zs = \prpv_\Zs(\lambda)$ protocol as follows:
\begin{enumerate}
  \item Starting at $t = 0$, $V_0$ samples a pair of keys $(\pk, \sk) \gets \keygen(1^\lambda)$, broadcasts $\pk$, and waits to receive $y_0$ from the prover before time $t < 4$, and $\ans_0$ at time $t = 4$.
  \item At $t = 1$, $V_1$ samples a uniformly random bit $b$ and broadcasts it, and waits to receive $y_1$ from the prover at time $t = 3$, and $\ans_1$ at time $t \le 5$.
  \item At $t = p_P$, the prover, located at $p_P \in [1, 2)$, receives $pk$, it prepares $(y, \rho) \gets \obligate(\pk)$ and broadcasts $y$.
  
  \item At $t = 4 - p_P$, the prover receives $b$, computes $\ans \gets \solve(\pk, y, \rho, b)$ and broadcasts $\ans$.
  
  \item At $t \le 5$ when the verifiers receive all the messages in time described above, they check that $y_0 = y_1$ and $\ans_0 = \ans_1$, and the answers pass the test: 
  \begin{align*}
      \ver(\sk, y_0, b, \ans_0) = 1 \,\wedge\, \ver(\sk, y_1, b, \ans_1) = 1. 
  \end{align*}
\end{enumerate}
\end{construction}

\begin{proposition}
  \label{thm:prpv-completeness}
  If $\Zs$ have completeness $c$, $\prpv_\Zs$ has position-robust completeness $c$ with possible prover locations $[1, 2)$. %
\end{proposition}
\begin{proof}
  For any prover location $p_P \in [1, 2)$, his first broadcast $y$ would be received by $V_0$ at time $2p_P < 4$, and by $V_1$ at time exactly 3; his second broadcast $\ans$ would be received by $V_0$ at time exactly 4, and by $V_1$ at time $1 + 2(3 - p_P) = 7 - 2p_P \le 5$.
  Therefore, the verifiers should then accept with probability $c$ by completeness of $\Zs$.
\end{proof}

Before considering soundness, we first justify that any arbitrary strategy could be compiled into a strategy where there are only two adversaries $A_0, A_1$ being at locations $0$ and $3$ respectively.

\begin{claim}
  \label{claim:adversarial-behavior}
  Let $S \in \mathcal R_P$ be any polynomially bounded adversarial strategy for $\prpv$.
  There exists another strategy $S' \in \mathcal R_P$, where for all $\lambda$, $S'_\lambda$ consists of only two adversaries at location $0$ and $3$ respectively, and $S'$ has the same verifier acceptance probability as $S$.
  Furthermore, the amount of entanglement shared between two adversaries in $S'$ is the same as the entanglement shared between adversaries in $(-\infty, 1)$ and those in $(2, +\infty)$ in $S$.
\end{claim}
\begin{proof}
  Without loss of generality, we assume that there is no adversaries in $(-\infty, 0) \cup (3, +\infty)$, since this is beyond where the verifiers sit and those adversaries can always be simulated by an adversary sitting at location $0$ and $3$.
  We look at all the adversaries in $[0, 1)$ (calling them ``the lefties'') and those in $(2, 3]$ (calling them ``the righties'').
  We focus on when their private messages crosses the ``event horizon'' $[1, 2]$ (calling them CPC, short for ``cross party communications'').
  Observe that any communication that is computed not in the two light cones\footnote{
    A light cone above $(x, t)$ includes all points in the space-time diagram $(x', t')$ such that $t' \geq t$ and $||x - x'|| \leq t' -t$.
    Similarly, a light cone below $(x, t)$ includes all points in the space-time diagram $(x', t')$ such that $t' \leq t$ and $||x - x'|| \leq t - t'$.
  } above $(0, 0)$ and $(3, 1)$ does not depend on either verifier's message, and therefore can be precomputed at time $t = -\infty$.
  Similarly, we can also discard any communication that happened outside of the two light cones below $(0, 4)$ and $(3, 3)$, since $V_0$ and $V_1$'s decisions only depends on information that comes under those two light cones.
  Therefore, the only meaningful message from the lefties needs to reach location 1 at time 1, and the only meaningful message from the righties need to reach location 2 at time 2.
  
  We can now see that a single adversary can be placed at location 0 that simulates all the lefties.
  In particular, he will at time 0, compute the actions for lefties on line $(0, 0) - (1, 1)$ in order, send the messages to the right, and simulate all the lefties' internal communications afterwards before $(3, 1) - (4, 0)$; at time 4, he receives the only meaningful message from the righties and compute the actions for lefties on line $(3, 1) - (4, 0)$ in order, and output the corresponding output.
  Using the same argument, a single simulated adversary at location 3 can also simulate the righties perfectly.
  Finally, since $S \in \mathcal R_P$, there are only polynomial number of adversaries from $S$, therefore the final adversaries we construct are also polynomial time.
\end{proof}

We now more formally characterize the behaviors of the two adversaries $A_0, A_1$ at position $0$ and $3$ respectively, as described above (also illustrated at \Cref{fig:prpv-protocol}).
\begin{enumerate}
  \setcounter{enumi}{-1}
  \item At $t = -\infty$, $A_0$ receives the security parameter $\lambda$, runs a set up $U_0$ to prepare state $\rho^{(0)}_{RS}$. It then sends $\rho^{(0)}_S$ to $A_1$. 
  
  Thus, we can assume before the protocol begins, $A_0$ and $A_1$ possess $\rho^{(0)}_R, \rho^{(0)}_S$ respectively. 
  
  \item At $t = 0$, $A_0$ receives $\pk$. It applies a quantum circuit $U_1$ on $\ket\pk\bra\pk$ and $\rho^{(0)}_R$ to get a classical string $y_0$ along with the resulting overall state $\rho^{(1)}_{R'MS}$ and sends the classical $y_0$ to $V_0$ immediately at time $t = 0$ (he could delay the message, but since he does not obtain new information before he has to send $y_0$ (at time $t < 4$), this does not help him).
  
  For the residual registers $R'$ and $M$, $A_0$ also sends $\rho^{(1)}_M$ to $A_1$ immediately at time $t = 0$, and stores $\rho^{(1)}_{R'}$.
  
  \item At $t = 1$, $A_1$ receives $b$, applies $U_2$ on $\ket b\bra b$ and $\rho^{(1)}_{S}$, and let the overall resulting state be $\rho^{(2)}_{R'MM'S'}$.
  He sends $\rho^{(2)}_{M'}$ to $A_0$ and stores $\rho^{(2)}_{S'}$.
  
  \item At $t = 3$, $A_1$ receives register $M$, performs a POVM measurement $U_3$ on $\rho^{(2)}_{MS'}$ to obtain $y_1, \ans_1$ to send to $V_1$.
  \item At $t = 4$, $A_0$ receives register $M'$, performs a POVM measurement $U_4$ on $\rho^{(2)}_{R'M'}$ to obtain $\ans_0$ to send to $V_0$.
\end{enumerate}

We now consider a special case of $U_2$, denoted by the challenge forwarding unitary $F$, where it on input $\ket b\bra b \otimes \rho_S^{(1)}$, and outputs $\ket b\bra b$ into register $M'$ and $\ket b\bra b \otimes \rho_S^{(1)}$ into register $S'$.
\begin{definition}
We denote the set of adversarial strategies having $U_2 = F$ to be $\mathcal R_F \subseteq \mathcal R_P$.
\end{definition}  

Note that for adversaries in $\mathcal R_F$ \emph{can} pre-share entanglement.
This is needed later in \Cref{cor:cvpv-qrom}.

\begin{theorem}
  \label{thm:prpv-sound-forwarding}
  If $\Zs$ have \oneoftwo non-local soundness $\tau$, $\prpv_\Zs$ has soundness $\tau$ against $\mathcal R_F$.
\end{theorem}
\begin{proof}
  The idea is to reduce this adversary to the 1-of-2 non-local game defined in \Cref{def:non-local_game}.
  Note that since $U_2 = F$ is the challenge forwarding operator, we can assume the challenge $b$ is given to both $A_0, A_1$ at time $t = 3, t = 4$ respectively, instead of being sent from $V_1$, which achieves the same acceptance probability.
  
  Assume for contradiction that $A_0, A_1$ achieves success probability noticeably more than $\tau$.
  We construct an adversary that breaks the \oneoftwo non-local soundness.
  \begin{itemize}
      \item The algorithm $\As$ upon receiving $\pk$, it runs $U_1 U_0$ on $\pk$ to obtain $y_0$ along with $\rho_{R'MS}^{(1)}$.
      It outputs $\rho_{R'}^{(1)}$ into register $B$ and $\rho_{MS}^{(1)}$ into register $C$.
      \item Since $M'$ register simply contains $b$, let $\Bs$ perform POVM $U_4$ as $M'$ at time $t = 4$, and let $\Cs$ perform POVM $U_3$ .
  \end{itemize}
  Thus, $(\As, \Bs, \Cs)$ perfectly simulates $(A_0, A_1)$ in $\prpv$, and thus achieves the same success probability in the non-local game as that by $(A_0, A_1)$ in $\prpv$, which by assumption violates the \oneoftwo non-local soundness $\tau$ of the underlying puzzle $\Zs$.
\end{proof}

We now show that this protocol is also sound with respect to a weaker restriction on the adversaries, in particular, the set of strategies where they do not pre-share entanglement.
\begin{theorem}
  \label{thm:prpv-sound-no-entanglement}
  If $\Zs$ have \oneoftwo non-local soundness $\tau$, $\prpv_\Zs$ has soundness $\tau$ against $\mathcal R_0$, i.e. the set of adversarial strategies where $S$ is simply a classical string.
\end{theorem}
\begin{proof}
  The idea is to compile any such adversary into a challenge-forwarding adversary, and invoke the soundness argument of \Cref{thm:prpv-sound-forwarding}.
  The equivalent adversary could be constructed as follows.
  The key observation is that $U_2$'s input only depends on $b$ which can only be one of two values, and register $S$ which only holds a classical string.
  Thus, $A_0$ can perform two executions of $U_2$ on $b = 0, 1$ respectively before it knows the actual $b$.
  \begin{enumerate}
    \setcounter{enumi}{-1}
    \item At $t = -\infty$, no set up is done.
    \item At $t = 0$, $A_0$ runs $U_1 U_0$ on $\pk$, in addition he also copies the classical string in register $S$, and runs $U_2$ (originally run by $A_1$) twice on two copies but for $b = 0, 1$ respectively.
    We denote the output as $\rho_{M'_0 S'_0}^{(2)} \otimes \rho_{M'_1 S'_1}^{(2)}$.
    He sends registers $S'_0, S'_1, M$ as his message and keep registers $R', M'_0, M'_1$ to himself.
    \item At $t = 1$, $A_1$ simply copies and forwards $b$, i.e. invokes $F$.
    \item At $t = 3$, $A_1$ receives registers $S'_0, S'_1, M$ and $b$, and runs the original POVM on registers $M, S'_b$.
    \item At $t = 4$, $A_0$ receives the new second private message $b$, and runs the original POVM on registers $R', M'_b$.
  \end{enumerate}
  It is easy to see that this compiler preserves the behavior of the adversaries perfectly, and therefore the soundness follows.
\end{proof}

\section{Parallel Repetition}

\subsection{Strong 1-of-2 Puzzles}

For the sake of the proof later, we identify two special properties from the underlying \oneoftwo puzzle.

\begin{definition}
\label{ver_pk_property}
  A \oneoftwo puzzle satisfies \emph{0-challenge-public-verifiability}, if $\ver$ can be computed using only $\pk$ when $b = 0$, i.e. for $(\pk, \sk) \gets \keygen(1^\lambda)$ and any $y, \ans$, $\ver(\sk, y, 0, \ans)$ can be efficiently computed correctly with probability 1 using only $(\pk, y, \ans)$.
\end{definition}

We can easily verify this property by looking at the NTCF construction, as we show in \Cref{cor:zero_pub_ver}.

The other property that we need later is that the \oneoftwo non-local soundness should be \textonehalf{} instead of \textthreequarters.
In order to achieve this, we recall the following theorem.

\begin{theorem}[{\cite[Corollary 2.10]{radian2020semi}}]
  \label{thm:strong1o2puzzle-base}
  Assuming the quantum hardness of LWE, there exists a $(1 - \negl, 0)$-\oneoftwo puzzle $\Zs$, furthermore, it satisfies 0-challenge-public-verifiability.
\end{theorem}

In the prior work, they call this a \emph{strong} \oneoftwo puzzle, and the proof is via constructing a flavor of parallel repetition of \emph{any} base $(1 - \negl, \frac12)$-\oneoftwo puzzle.

Invoking \Cref{thm:2of2vsnon-local}, we get the following.
\begin{lemma}
  \label{lem:strong1o2puzzle}
  Assuming the quantum hardness of LWE, there exists a \oneoftwo puzzle with completeness $1 - \negl$, \oneoftwo non-local soundness \textonehalf, and 0-challenge-public-verifiability.
\end{lemma}

\subsection{\oneofpoweroftwo Puzzles}

Unfortunately, for any \oneoftwo puzzle $\Zs$ with completeness $c$, there is always an adversarial strategy breaking $\prpv_\Zs$ (and therefore non-local soundness for $\Zs$) with probability $c/2$ by simply guessing the challenge $b$, which would be correct with probability \textonehalf.
Therefore, in order to beat this barrier, we consider the parallel repetition of \oneoftwo puzzles, and show that non-local soundness decreases exponentially, as a stepping stone to achieving negligibly-sound position verification.

\begin{definition}[\oneofpoweroftwo Puzzles]
  A \oneofpoweroftwo puzzle is exactly the same as a \oneoftwo puzzle, except that the challenge $b$ is a uniformly random $k$-bit bitstring instead of a single random bit.
\end{definition}

Two main requirements of interest for \oneofpoweroftwo puzzles are completeness and \oneofpoweroftwo non-local soundness for \oneofpoweroftwo puzzles, which can naturally be extended from completeness and \oneoftwo non-local soundness (see \Cref{def:non-local_game}, \Cref{def:non-local-soundness}) for \oneoftwo puzzles by simply changing $b$.

\begin{construction}[Parallel Repetition of \oneoftwo Puzzles]
  Let $\Zs = (\keygen, \obligate, \solve, \ver)$ be a \oneoftwo puzzle.
  The $k$-fold parallel repetition of $\Zs$, denoted as $\Zs^k$, is a \oneofpoweroftwo puzzle constructed as follows:
  \begin{itemize}
    \item $\keygen, \obligate, \solve$ for $\Zs^k$ simply runs $\keygen, \obligate, \solve$ for $\Zs$ $k$ times respectively.
    \item $\ver$ for $\Zs^k$ runs $\ver$ on all $k$ instances, and accepts if and only if all of them accept.
  \end{itemize}
\end{construction}

We note that the parallel repetition we consider here is different from the one used in the proof of \Cref{thm:strong1o2puzzle-base}.
In their case, the same challenge bit is reused across all $k$ instances, whereas in our case, each instance has a fresh random bit.
Thus, they obtain a \oneoftwo puzzle after the repetition, and here we get a \oneofpoweroftwo puzzle.

\begin{theorem}
  \label{thm:parallel-rep-1of2m}
  Let $\Zs$ be a \oneoftwo puzzle with completeness $1 - \negl$, \oneoftwo non-local soundness \textonehalf, and 0-challenge-public-verifiability.
  Then for any $k = \poly(\lambda)$, $\Zs^k$ has completeness $1 - \negl$ and \oneofpoweroftwo non-local soundness $2^{-k}$.
\end{theorem}

Recall that by definition, this theorem means that for any QPT adversary, there exists a negligible function $\negl$, such that the success probability is at most $2^{-k} + \negl(\lambda)$.

The rest of the section will be dedicated to proving the theorem, mainly the non-local soundness.
Our proof follows the ideas from the work by Alagic et al.\,\cite[Section 4]{AlagicCGH20}.
This work along with the one by Chia et al.\,\cite{chia2019classical} proved that parallel repetition decreases soundness exponentially for Mahadev's quantum delegation, which is a different application of NTCFs.

We begin by showing that for $\Zs$, the projection corresponding to challenge $0$ and $1$ are ``computationally orthogonal''.

\begin{lemma}
  \label{lem:perfect-0-player}
  Let $\Zs$ be a \oneoftwo puzzle with \oneoftwo non-local soundness \textonehalf.
  For any efficient non-local player $\Ws = (\As, \Bs, \Cs)$, let $\ket{\psi_{\pk}}_{BCY}$ be the overall state\footnote{Without loss of generality, we assume the state is purified where the auxiliary space is either in $B$ or $C$ register.} of $\As$'s output, where measuring $Y$ register in the computational basis gives $y$.
  Let projection $\Pi_{\sk, b}$ correspond to the quantum predicate $\ver(\sk, y, b, \ans_\Bs) \wedge \ver(\sk, y, b, \ans_\Cs)$.
  If
  \[\E_{\keygen}\left[\braket{\psi_\pk | \Pi_{\sk, 0} | \psi_\pk}\right] \ge 1 - \negl(\lambda),\]
  then
  \[\E_{\keygen}\left[\braket{\psi_\pk | \Pi_{\sk, 1} | \psi_\pk}\right] = \negl(\lambda).\]
\end{lemma}
\begin{proof}
  Note that by definition, $\Pi_{\sk, b}$ only involves running $\Bs, \Cs$ and $\ver$.
  Since $\Bs, \Cs$ does not act on the $Y$ register, and $\ver$ is only classically controlled on $Y$, $\Pi_{\sk, b}$ commutes with the computational-basis measurement on $Y$ for any $\sk, b$.
  Therefore, the overall winning probability for $\Ws$ is exactly $\frac12 \E_{\keygen}\left[\braket{\psi_\pk | (\Pi_{\sk, 0} + \Pi_{\sk, 1}) | \psi_\pk}\right]$.
  Suppose the lemma does not hold, then the expression above would be noticeably higher than \textonehalf.
  A contradiction.
\end{proof}

We now show that a similar version also holds for $\Zs^k$.
The notation $\Pi_{\sk, b}$ for \oneoftwo puzzles above can be extended similarly for \oneofpoweroftwo puzzles.

\begin{lemma}
  \label{lem:abba-is-negl}
  Let $\Zs$ be a \oneoftwo puzzle with completeness $1 - \negl$, \oneoftwo non-local soundness \textonehalf, and 0-challenge-public-verifiability.
  For any $k = \poly(\lambda)$, any non-local player $\Ws$ for $\Zs^k$, and any $a \neq b \in \{0, 1\}^k$,
  \[\E_{\keygen}\left[\braket{\psi_\pk | (\Pi_{\sk, b}\Pi_{\sk, a} + \Pi_{\sk, a}\Pi_{\sk, b}) | \psi_\pk}\right] = \negl(\lambda).\]
\end{lemma}
\begin{proof}
  Since $a, b$ is symmetric, without loss of generality, assume that there exists $i$ so that $a_i = 1$ and $b_i = 0$. 
  We claim that
  \begin{equation}
    \label{eq:projection-orthogonality}
    \E_{\keygen}\left[\braket{\psi_\pk | \Pi_{\sk, b}\Pi_{\sk, a}\Pi_{\sk, b} | \psi_\pk}\right] = \negl(\lambda).
  \end{equation}
  Assume for contradiction that there is a player strategy that makes this term noticeable, say $1/\eta(\lambda)$ for some polynomial $\eta$.
  Define $\Sigma_{\sk, a}$ projector to be same as $\Pi_{\sk, a}$ except that it only checks the $i$-th repetition and acts as identity on every other repetition.
  By definition $\Pi_{\sk, a} \preccurlyeq \Sigma_{\sk, a}$, and thus $\E_{\keygen}\left[\braket{\psi_\pk | \Pi_{\sk, b}\Sigma_{\sk, a}\Pi_{\sk, b} | \psi_\pk}\right] \ge 1/\eta$.
  Consider the single-copy non-local player $\Ws^* = (\As^*, \Bs^*, \Cs^*)$ as follows:
  \begin{itemize}
      \item For $\As^*$, on input $\pk^*$, it runs $\keygen$ $(k - 1)$ times, and set $\pk$ to be a list of $k$ public keys with $\pk^*$ inserted at the $i$-th position.
      \item Repeat the following for at most $q := \max\{4\eta^2, \lambda^2\}$ times: prepare $\ket{\psi_{\pk}}_{BCY}$, apply measurement $\Pi_{\sk, b}$, and abort the loop if the measurement accepts.
      This is efficient as we know all the secret keys except for index $i$, which we can still verify by 0-challenge-public-verifiability as $b_i = 0$.
      
      If the loop has not succeeded after $q$ iterations, $\As$ simply invokes the honest prover, guesses the challenge is 0, and sends the $\ans$ corresponding to 0-challenge to $\Bs$ and $\Cs$, so that later they simply output $\ans$.
      \item Measure $Y$ register of the residual state as $y$, and send $B, C$ registers as input registers for the next round.
      Let the overall residual state be denoted as $\ket\phi$.
      \item For $\Bs^*, \Cs^*$, if the challenge is 0 (or 1), run $\Bs, \Cs$ on $b$ (or $a$ respectively), and output index $i$.
  \end{itemize}
  By construction, the probability that $\Ws$ succeeds when challenge is 0 is $1 - \negl$.
  
  Now consider the case for challenge 1.
  Let
  \[\Omega := \left\{(\pk, \sk) : \braket{\psi_\pk | \Pi_{\sk, b} | \psi_\pk} > q^{-1/2} \right\}\]
  denote the set of ``good'' keys for the parallel scheme.
  For $(\pk, \sk) \in \Omega$, the probability of not terminating within $q$ iterations is at most $(1 - q^{-1/2})^q \le e^{-\sqrt q} \le e^{-\lambda}$ and thus $\ket\phi$ is exponentially close to $\Pi_{\sk, b}\ket{\psi_\pk}$.
  By \Cref{lem:perfect-0-player}, we have
  \[\E_{\keygen \wedge \Omega}\left[\braket{\phi | \Sigma_{\sk, a} | \phi}\right] \le \E_{\keygen}\left[\braket{\phi | \Sigma_{\sk, a} | \phi}\right] = \negl(\lambda),\]
  where $\E_{\keygen \land \Omega}[f(\pk, \sk)] = \E_{\keygen}[f(\pk, \sk) \cdot \mathbf{1}_{(\pk, \sk) \in \Omega}])$.  
  However,
  \begin{align*}
    \eta^{-1} &\le \E_{\keygen}\left[\braket{\psi_\pk | \Pi_{\sk, b}\Sigma_{\sk, a}\Pi_{\sk, b} | \psi_\pk}\right] \\
      &= \E_{\keygen \land \Omega}\left[\braket{\psi_\pk | \Pi_{\sk, b}\Sigma_{\sk, a}\Pi_{\sk, b} | \psi_\pk}\right] + \E_{\keygen \land \lnot\Omega}\left[\braket{\psi_\pk | \Pi_{\sk, b}\Sigma_{\sk, a}\Pi_{\sk, b} | \psi_\pk}\right] \\
      &\le \negl(\lambda) + q^{-1/2} \\
      &\le \negl(\lambda) + \eta^{-1}/2,
  \end{align*}
  which contradicts with that $\eta^{-1}$ is noticeable.
  Therefore, \eqref{eq:projection-orthogonality} must hold.
  It then follows that
  \begin{align*}
    \E_{\keygen}\left[\braket{\psi_\pk | (\Pi_{\sk, b}\Pi_{\sk, a} + \Pi_{\sk, a}\Pi_{\sk, b}) | \psi_\pk}\right]
      &= 2\E_{\keygen}\left[\Re\left(\braket{\psi_\pk | \Pi_{\sk, b}\Pi_{\sk, a} | \psi_\pk}\right)\right] \\
      &\le 2\E_{\keygen}\left[\left|\braket{\psi_\pk | \Pi_{\sk, b}\Pi_{\sk, a} | \psi_\pk}\right|\right] \\
      &\le 2\E_{\keygen}\left[\braket{\psi_\pk | \Pi_{\sk, b}\Pi_{\sk, a}\Pi_{\sk, b} | \psi_\pk}^{1/2}\right] \\
      &\le 2\E_{\keygen}\left[\braket{\psi_\pk | \Pi_{\sk, b}\Pi_{\sk, a}\Pi_{\sk, b} | \psi_\pk}\right]^{1/2} \\
      &\le \negl(\lambda)
  \end{align*}
  as claimed.
\end{proof}

We now recall a key technical lemma from the prior work before proving the main theorem.

\begin{lemma}[{\cite[Lemma 4.3]{AlagicCGH20}}]
  \label{lem:acgh-key-technical}
  Let $A_1, ..., A_m$ be projectors and $\ket\psi$ be a quantum state.
  Suppose there are real numbers $\delta_{ij} \in [0, 2]$ such that $\braket{\psi|(A_iA_j + A_jA_i)|\psi} \le \delta_{ij}$ for all $i \neq j$.
  Then $\braket{\psi|(A_1 + \cdots + A_m) |\psi} \le 1 + \left(\sum_{i < j} \delta_{ij}\right)^{1/2}$.
\end{lemma}

\begin{proof}[Proof of \Cref{thm:parallel-rep-1of2m}]
  Completeness follows immediately by running the honest \oneoftwo solver in parallel and a union bound on the failure probability.

  For soundness, consider any $1$-of-$2^k$ non-local player $\Ws$, and let its success probability be
  \[\tau = 2^{-k} \E_{\keygen} \left[\braket{\psi_\pk| (\sum_{c} \Pi_{\sk, c}) |\psi_\pk}\right]. \]
  Define an arbitrary total order ``$<$'' on $\{0, 1\}^k$.
  Then by \Cref{lem:acgh-key-technical} and \Cref{lem:abba-is-negl}, we have
  \begin{align*}
    \tau &\le 2^{-k} + 2^{-k}\E_\keygen \left[\braket{\psi_\pk| \sum_{a < b} (\Pi_{\sk, a}\Pi_{\sk, b} + \Pi_{\sk, b}\Pi_{\sk, a}) |\psi_\pk}\right]^{1/2} \\
      &\le 2^{-k} + 2^{-k} \sqrt{2^{2k} \cdot \negl(\lambda)} \\
      &= 2^{-k} + \negl(\lambda),
  \end{align*}
  concluding the proof.
\end{proof}

\subsection{0-Entanglement Soundness from Polynomial Hardness}

We now consider the $k$-fold pararell repetition of $\prpv$.

\begin{construction}
  Let $\Zs$ be a \oneoftwo puzzle.
  $\prpv_\Zs^k$ is the position verification protocol where the two verifiers and the prover runs $k$ instances of $\prpv_\Zs$ in parallel.
  At the end, the verifiers accept if and only if all $k$ instances accept.
\end{construction}

One can also naturally define $\prpv_{\Zs'}$ for any \oneofpoweroftwo puzzle $\Zs'$ by simply changing $b$ to be a $k$-bit bitstring.
By construction, $\prpv_{\Zs^k}$ results in the exact same protocol as $\prpv_\Zs^k$ for any \oneoftwo puzzle $\Zs$.

\begin{theorem}
  \label{thm:prpv-negl}
  Let $\Zs$ be a \oneoftwo puzzle with completeness $1 - \negl$ and \oneoftwo non-local soundness $1/2$.
  Then for any $k = \poly(\lambda)$, $\prpv_\Zs^k$ has position-robust completeness $1 - \negl$ with possible prover locations $[1, 2)$, and soundness $2^{-k}$ against both $\mathcal R_F$ and $\mathcal R_0$.
\end{theorem}
\begin{proof}
  Completeness follows directly by invoking the honest prover in parallel.
  Since the protocol has the exact same timing constraints as before, the adversarial behavior can be described exactly the same as \Cref{claim:adversarial-behavior}.
  Therefore, soundness against $\mathcal R_F$ can be proven exactly the same as \Cref{thm:prpv-sound-forwarding} as the proof does not rely on $b$ being a single bit, except that we instead reduce to $1$-of-$2^k$ non-local soundness of $\Zs^k$, which we prove in \Cref{thm:parallel-rep-1of2m}.
  
  As for soundness against $\mathcal R_0$, we also employ the same idea, which is to compile any such adversary into a challenge forwarding adversary, and invoke the soundness against $\mathcal R_F$.
  Note that the proof of \Cref{thm:prpv-sound-no-entanglement} requires running $U_2$ on all possible challenges $b$, and in this case, the challenge space is as large as $2^k$.
  The compiler from \Cref{thm:prpv-sound-no-entanglement} trying all possible coins will give a forwarding adversarial strategy with run-time/communication blow up in $2^k$.
  Therefore, the proof immediately extends if $k = O(\log \lambda)$.
  
  Assume $k = \omega(\log \lambda)$, then we need to prove that any efficient adversary strategy can only succeed with probability at most $2^{-k} + \negl(\lambda)$, which is overall a negligible function.
  Assume if an $\mathcal R_0$ adversary is able to break the protocol with probability $1/p$ for some polynomial $p(\lambda)$, then we can come up with an adversary that breaks $\prpv_\Zs^m$ with the same probability $1/p$ for any $m \le k$, by simply simulating the other $(m - k)$ executions and running the original adversary.
  Pick $m = \log p + 1 = O(\log \lambda)$, and we get an $\mathcal R_F$ adversary for $\prpv_\Zs^m$ with success probability $1/p$.
  However, as we argued, $\prpv_\Zs^m$ has soundness $2^{-m} = 2^{1 - \log p} = 1/2p$ against $\mathcal R_F$.
  This leads to a contradiction as $1/2p$ is noticeable.
\end{proof}

Combining this with \Cref{lem:strong1o2puzzle}, we get the following.
\begin{corollary}
  \label{cor:cvpv}
  Assuming polynomial quantum hardness of LWE, there is a classically-verifiable position verification scheme having position-robust completeness $1 - \negl$ with possible prover locations $[1, 2)$ and negligible soundness against $\mathcal{R}_F$ and $\mathcal R_0$.
\end{corollary}

\subsection{Bounded-Entanglement Soundness from Subexponential Hardness}

The soundness could be decreased further if we assume stronger hardness on quantum LWE.
For $c > 0$ and $L: \mathbb N^+ \to \mathbb N$, let $\mathcal R_{L, c}$ be the same as $\mathcal R_L$, except that the adversaries can run in time $\poly\left(2^{\lambda^c}\right)$ instead of $\poly(\lambda)$.
We denote $c$-subexponential hardness\,\cite{jawale2021snargs, holmgren2021fiat} to mean that any $\poly\left(2^{\lambda^c}\right)$-time adversary achieving advantage at most $\negl\left(2^{\lambda^c}\right)$. 

\begin{lemma}
  Assuming $c$-subexponential quantum hardness of LWE, there is a classically-verifiable position verification scheme having position-robust completeness $1 - \negl$ with possible prover locations $[1, 2)$ and $c$-subexponential soundness against $\mathcal R_{0, c}$.
\end{lemma}
\begin{proof}
    Most of the reduction extends immediately as they are all black-box reductions, in particular, they do not explicitly depend on the adversary's running time nor its success probability.
    There are three exceptions: \Cref{thm:strong1o2puzzle-base}, \Cref{thm:parallel-rep-1of2m} and \Cref{thm:prpv-negl}.
    
    We show how to adapt the proof for \Cref{thm:parallel-rep-1of2m} for the subexponential case, and the same approach could be applied to \Cref{thm:strong1o2puzzle-base} as well, which corresponds to the parallel repetition of the base \oneoftwo puzzles.
    The non-black-box reduction in the proof for \Cref{thm:parallel-rep-1of2m} occured in \Cref{lem:abba-is-negl}, where the reduction needs to repeat running the adversary $q$ times, where $q = \max \{\lambda^2, (2/p)^2\}$ and $p$ is the adversary's success probability in \eqref{eq:projection-orthogonality}.
    By assumption, there exists some constant $c > 0$, such that for any $\poly\left(2^{\lambda^c}\right)$-time adversary, he can win the \oneoftwo non-local game with probability at most $1/2 + \negl\left(2^{\lambda^c}\right)$.
    Let the adversary's running time be $T = \poly\left(2^{\lambda^c}\right)$ and assume that $p = 1/\poly\left(2^{\lambda^c}\right)$, then we can see that the reduction still runs in time $O(Tq) = O(T/p^2) = \poly\left(2^{\lambda^c}\right)$.
    The rest of the proof goes through and at the end, we conclude that $p = \negl\left(2^{\lambda^c}\right)$, a contradiction.
    
    The non-black-box step in the proof for \Cref{thm:prpv-negl} occurred in the step where the reduction reduces a general adversary into a challenge forwarding adversary.
    Recall that $k$ is the number of repetitions.
    Again, if $k = O(\lambda^c)$, the reduction runs in time $2^k \cdot \poly\left(2^{\lambda^c}\right) = \poly\left(2^{\lambda^c}\right)$ and the proof goes through.
    Assume $k = \omega(\lambda^c)$ and the adversary's success probability $p$ is noticeably larger than $2^{-k}$ which is again negligible in $2^{\lambda^c}$, then the reduction runs in time $O(\log 1/p) \cdot \poly\left(2^{\lambda^c}\right) = \poly\left(2^{\lambda^c}\right)$ and the proof also goes through.
    
    Finally, we remark that the ``statistical security parameter'' in the instantiation of NTCFs, which is the ratio $\frac{B_P}{B_V} = \frac{B_V}{B_L}$, needs to be set to be at least $2^{\lambda^{c'}}$ for any $c' > c$ in order for the NTCF construction to be $c$-subexponentially secure.
    We refer the readers to \cite[Remark 4.2]{brakerski2021cryptographic} for the relevant discussions on the choice of the parameters for the NTCF.
\end{proof}

We now leverage standard techniques~\cite{aaronson2004limitations,Tomamichel_2013} to show that this can be bootstrapped to handle bounded entanglement adversaries.

\begin{theorem}
  \label{thm:bounded-entanglement}
  Assuming $c$-subexponential quantum hardness of LWE, for any polynomial $L: \mathbb N^+ \to \mathbb N$, there is a classically-verifiable position verification scheme having position-robust completeness $1 - \negl$ with possible prover locations $[1, 2)$ and $c$-subexponential soundness against $\mathcal R_{L, c}$.
\end{theorem}
\begin{proof}
  By assumption, there exists an \oneoftwo puzzle $\Zs$ such that for any polynomial $k(n) = \omega(n^c)$, for any adversarial strategy $S \in \mathcal R_{0, c}$, $S$ succeeds in breaking $\prpv_\Zs^k$ with probability at most $\negl\left(2^{n^c}\right) = 2^{-\omega(n^c)}$.
  In particular, $S$ has success probability at most $2^{-n^c}$ for all sufficiently large $n$ against $\mathcal R_{0, c}$.
  
  Pick $n = \left(L(\lambda) + \lambda^{c}\right)^{1/c}$, and we claim that for every adversarial strategy $S \in \mathcal R_{L, c}$, its success probability in breaking $\prpv_\Zs^k(n)$ is $\negl\left(2^{-\lambda^{c}}\right)$, in particular, it is at most $2^{-\lambda^{c}}$ for all sufficiently large $\lambda$.
  Assume this is not the case, then its success probability is higher than $2^{-\lambda^{c}}$ infinitely often.
  We can replace the entanglement with maximally mixed state, and therefore we have an adversarial strategy $S' \in \mathcal R_{0, c}$ with success probability higher than $2^{-\lambda^{c}} \cdot 2^{-L} = 2^{-L - \lambda^{c}} = 2^{-n^{c}}$ infinitely often.
  This is because any pre-shared entanglement of dimension $2^L$ can be extended into a basis for a $2^L$-dimension state space, which means that the entanglement can be replaced by a maximally mixed state (a non-entangled state), and reduces the probability by at most a factor $2^L$. 
  This contradicts with the conclusion in the last paragraph.
\end{proof}

\section{Attacks and Countermeasures}
\subsection{Attack with Polynomial Entanglement}
\label{sec:attack}

In this section, we present an adversarial strategy in $\mathcal R_L$ for $\prpv_\Zs^k$ achieving winning probability as good as the completeness of the protocol for any $k = \poly(\lambda)$ and any \oneoftwo puzzle $\Zs$ satisfying a specific property defined below, with $L$ being only as large as the number of qubits in $\rho$.
Note that this does not contradict the bounded-entanglement soundness from \Cref{thm:bounded-entanglement}, as there $L$ is determined by the security parameter $\lambda$, and therefore cannot depend on the specific protocol construction (and therefore the length of $\rho$).

\begin{definition}
  We call a \oneoftwo puzzle having an \emph{XZ-solver} if $\solve(\pk, y, \rho, b)$ simply measures $\rho$ (as a string of qubits) in standard basis if $b = 0$, or Hadamard basis if $b = 1$, and outputs the measurement result.
\end{definition}
Note that the \oneoftwo puzzle based on NTCFs both have an XZ-solver (see \Cref{cor:XZ_measure}).
This holds also for the strong \oneoftwo puzzle, as by construction, it is simply running an XZ-solver several times in parallel.
We now describe the attack for $\prpv_\Zs$ when $\Zs$ has an XZ-solver.%

\begin{enumerate}
  \item At $t = -\infty$, $A_0$ prepares $L$ EPR pairs, and keeps the first half in register $R$ and sends the other half to $A_1$ in register $S$.
  \item At $t = 0$, $A_0$ receives $pk$ and prepares $(y, \rho)$ by running $\obligate$ as normal.
  In addition, he also teleports $\rho$ using $R$, getting measurement results $(k_0, k_1)$.
  He sends $y$ to $V_0$, and $(y, k_0, k_1)$ to $A_1$.
  \item At $t = 1$, $A_1$ receives $b$, and measures $S$ in standard basis if $b = 0$, or else in Hadamard basis if $b = 1$.
    He obtains measurement results $r$ and sends $(b, r)$ to $A_0$.
  \item At $t = 3$, $A_1$ receives $(y, k_0, k_1)$, and sends $(y, r \oplus k_b)$ to $V_1$.
  \item At $t = 4$, $A_0$ receives $(b, r)$, and sends $r \oplus k_b$ to $V_0$.
\end{enumerate}

The correctness of the attack follows from the fact that the teleportation gadget commutes with the $X$/$Z$ measurements the prover performs at the end of the computation.

In order to attack $\prpv_\Zs^k$, we can simply repeat this attack $k$ times in parallel.

\subsection{Unbounded-Entanglement Soundness in the QROM}

In this section, we employ the idea from Unruh's work~\cite{unruh2014qrom} for constructing position verification with quantum communication against unbounded entanglement in the quantum random oracle model (QROM).
Following the idea there, we change the construction $\prpv_{\Zs}^k$ so that the random $k$-bit challenge $b$ is sampled by the random oracle, i.e. $b = H(x_0 \oplus x_1)$ for random $x_0, x_1$.
We prove that the resulting construction achieve negligible soundness against efficient adversaries with \emph{any} polynomial amount of entanglement in the QROM.

The following construction $\prpvrom$ is a variant of $\prpv$ in the QROM.
We highlight the differences in this construction with \varul{underlines}.
\begin{construction}[$\prpvrom$ Protocol]
\label{construction:prpv_qrom}
Let $\Zs$ be an \oneofpoweroftwo puzzle and let $H : \{0,1\}^\lambda \to \{0,1\}^k$ be a random oracle, where $\lambda$ is the security parameter. 
The protocol $\prpvrom_\Zs$ is defined as follows: 
\begin{enumerate}
  \item Starting at $t = 0$, $V_0$ samples a pair of keys $(\pk, \sk) \gets \keygen(1^\lambda)$ \varul{and ${x_0 \gets \{0,1\}^\lambda}$}, broadcasts $\pk$ \varul{and $x_0$}, and waits to receive $y_0$ from the prover before time $t < 4$, and $\ans_0$ at time $t = 4$.
  
  \item At $t = 1$, $V_1$ samples \varul{$x_1 \gets \{0,1\}^\lambda$} and broadcasts it, and waits to receive $y_1$ from the prover at time $t = 3$, and $\ans_1$ at time $t \le 5$.
  
  \item At $t = p_P$, the prover, located at $p_P \in [1, 2)$, receives $pk$ \varul{and $x_0$}, it prepares $(y, \rho) \gets \obligate(\pk)$ and broadcasts $y$.
  
  \item At $t = 4 - p_P$, the prover receives  \varul{$x_1$}, \varul{let $b = H(x_0 \oplus x_1)$}; it computes $\ans \gets \solve(\pk, y, \rho, b)$ and broadcasts $\ans$.
  
  \item At $t \le 5$ when the verifiers receive all the messages in time described above, they check that $y_0 = y_1$ and $\ans_0 = \ans_1$, and the answers pass the test:
  
  \centerline{$\ver(\sk, y_0, b, \ans_0) = 1 \,\wedge\, \ver(\sk, y_1, b, \ans_1) = 1$, \quad \varul{where $b = H(x_0 \oplus x_1)$}.}
\end{enumerate}
\end{construction}

\begin{theorem}
For any $k > 0$ and any \oneoftwo puzzle $\Zs$, if $\prpv_\Zs^k$ has completeness $c$ and soundness $s$ against $\mathcal R_F$, $\prpvrom_{\Zs^k}$ has completeness $c$ and soundness $s$ against $\mathcal{R}_P$, i.e. the set of polynomially bounded strategies (and therefore having at most polynomial amount of pre-shared entanglement) in the QROM.
\end{theorem}
\begin{proof}
The completeness of the protocol extends since the only change is how $b$ is sampled.

Since the timing constraints remain the same, similar to \Cref{claim:adversarial-behavior}, any arbitrary strategy can be compiled into a strategy where there are only two adversaries $A_0, A_1$ being at locations $0$ and $3$ respectively. 
In particular, besides the set up $U_0$, $A_0$'s behavior can be characterized as an action $U_1$ at time $0$ when it receives $\pk, x_0$ from $V_0$ and a POVM measurement $U_4$ at $t = 4$; similarly, $A_1$'s behavior can be characterized as an action $U_2$ at time $1$ when it receives $\pk, x_1$ from $V_1$ and a POVM measurement $U_3$ at $t = 3$.

The proof is through a sequence of hybrid arguments.
We give the hybrids below, and show the success probability in each hybrid is negligibly close to the previous one.

\begin{description}
\item [Hybrid 0] Execute the original protocol with $V_0, V_1, A_0, A_1$ and a random oracle $H$.

\item [Hybrid 1] It is the same as Hybrid 0 except that at time $t = 2$, 
a random challenge $b \gets \{0,1\}$ is chosen randomly. Then the random oracle $H$ is immediately reprogrammed such that $H(x_0 \oplus x_1) = b$.

The indistinguishability comes from a variant of the one-way to hiding (O2H) lemma \cite[Lemma 3]{unruh2014qrom}.
Let $A_0, A_1$ be $T(\lambda)$-time bounded (with possible shared-entanglement), where $T = \poly(\lambda)$ since the adversaries are efficient.
With this O2H lemma, we can argue that Hybrid 1 and Hybrid 0 are $O(T(\lambda) \cdot 2^{-\lambda/2}) = \negl(\lambda)$ close. 

\item [Hybrid 2] We can without loss of generality assume $A_0, A_1$ get access to two but identical random oracles instead of one random oracle. The hybrid is the same as Hybrid 1 except that the random oracle $H$ accessed by $A_1$ is reprogrammed immediately before time $t = 3$, and the same random oracle $H$ accessed by $A_0$ is reprogrammed immediately before time $t = 4$.

After time $t > 2$, $A_0$ only does computation at time $t = 4$ and $A_1$ only does computation at time $t = 3$. Thus, the output distributions of Hybrid 1 and Hybrid 2 are identical.
\end{description}

Assume that the adversaries' success probability in breaking $\prpvrom_{\Zs^k}$ is $p = p(\lambda)$, then by the hybrid argument, they will also succeed in breaking Hybrid 2 with probability at least $p - \negl(\lambda)$.
We now consider compiling any adversary strategy in Hybrid 2 can be converted into a (challenge-forwarding) adversary strategy in $\mathcal{R}_F$ for $\prpv_{\Zs}^k$ with the same success probability.
\begin{itemize}
    \item Recall that $B_0$ is at $0$ and $B_1$ is at $3$.
    $B_0$ samples random $x_1$ and executes $U_2 U_0$ and obtain outputs in registers $R, M', S'$.
    $B_0$ also samples a random $x$ (and let $x_0 := x \oplus x_1$), and a random oracle $H$.
    Note that although the description of $H$ is inefficient, $H$ can be perfectly simulated with a $2 T$-wise independent function~\cite[Theorem 3.1]{C:Zhandry12}, whose description is efficient.
    
    $B_0$ possesses registers $R, M'$ and classical information $H, x$ and sends to $B_1$ register $S'$ and classical information $H, x$.
    \item At time $t = 0$, when $B_0$ gets $\pk$ from $V_0$, it runs $U_1$ on input $\pk, x_0$ and register $R$ with oracle access to $H$, and sends the resulting $y$ to $V_0$ and $M$ register to $B_1$.
    \item At time $t = 1$, $B_1$ simply runs $F$, i.e. forwards $b$ to $B_0$.
    \item At time $t = 3$, $B_1$ receives $M$ and perform the POVM on $M, S'$ with oracle access to $H_{x, b}$, where $H_{x, b}$ denotes the reprogrammed function $H$ that on input $x$ outputs $b$, and runs $H$ on input anything else.
    \item At time $t = 4$, $B_0$ receives $b$ and perform the POVM on $R', M'$ with oracle access to the reprogrammed $H_{x, b}$. 
\end{itemize}
We can see that this perfectly simulates the output distribution of Hybrid 2, and thus $p - \negl(\lambda) \le s + \negl(\lambda)$ as desired.
\end{proof}

Combining this with \Cref{cor:cvpv}, we get the following.
\begin{corollary}
  \label{cor:cvpv-qrom}
  Assuming polynomial quantum hardness of LWE, there is a classically-verifiable position verification scheme having position-robust completeness $1 - \negl$ with possible prover locations $[1, 2)$ and negligible soundness against $\mathcal{R}_P$ in the QROM.
\end{corollary}

\section{Necessity of Proofs of Quantumness}

In this section, we argue that the proof of quantumness is necessary to construct classically-verifiable position verification protocol even in one dimension.
We first recall the definition of proofs of quantumness.
The motivation is to test whether an untrusted efficient\footnote{Indeed, an unbounded classical device can always simulate the quantum strategy, and no tests can tell the difference.} device truly has quantum capabilities.

\begin{definition}
  A proof of quantumness is an interactive protocol with an efficient classical verifier satisfying: \begin{itemize}
      \item \textbf{Completeness $c$}: There exists a polynomial-time quantum prover that can convince the verifier with probability at least $c$;
      \item \textbf{Soundness $s$}: Any polynomial-time classical prover convinces the verifier with probability at most $s + \negl$ for some negligible function $\negl$.
  \end{itemize}
\end{definition}

\begin{theorem}
  \label{thm:necessity-poq}
  Assuming the existence of any 1D position verification protocol satifying: \begin{enumerate}
    \item Without loss of generality, there are two verifiers $V_0, V_1$ at location 0 and 1 respectively;
    \item It has completeness $c$ for an efficient prover at location $p_P \in (0, 1)$;
    \item It has soundness $s$ against two efficient classical adversaries, one located in $[0, p_P)$ and the other located in $(p_P, 1]$;
    \item Verifiers are classical and efficient.
  \end{enumerate}
  There exists a proof of quantumness protocol with completeness $c$ and soundness $s$.
\end{theorem}
\begin{proof}
  The construction for the proof of quantumness is very simple: it simply runs both $V_0, V_1$, sending and receiving the messages in the order enforced by the timing constraints for the prover in the position verification protocol.
  We emphasize that the resulting proof of quantumness protocol is a standard interactive protocol without timing constraints, since we are only using the timing constraints from position verification protocol to define the order of the messages.
  Therefore, completeness follows immediately; the verifier is efficient and classical as the position verification protocol is efficient and classically verifiable.
  
  For soundness, the idea is essentially extending the impossibility of classical position verification \cite{CGMO09PBC}.
  Assume there is a classical prover that wins this protocol with probability $w$.
  We construct an adversarial strategy with two classical adversaries, $A_0$ at location $p_0 \in L \cap [0, p_P)$, and $A_1$ at location $p_1 \in L \cap (p_P, 1]$. They sample and pre-share their random tape used to run the classical prover. In the security game, they receive the broadcast message from the verifiers, and $A_b$ will send the message computed to $V_b$ for $b = 0, 1$.
  
  To see that this adversary simulates the classical prover perfectly, it suffices to show that the adversaries can obtain all the information they need \emph{in time} to simulate the prover's responses.
  Consider any message computed by the prover at location $p_P$ at time $t = \tau$.
  It can only depend on messages sent from $V_0$ before $t \le \tau - p_P$, and from $V_1$ before $t \le \tau - (1 - p_P)$; this response will be received by $V_0$ at $t = \tau + p_P$ and by $V_1$ at $t = \tau + (1 - p_P)$.
  Therefore $A_0$ needs to simulate the response before $t = \tau + p_P - p_0$.
  Since $\tau + p_P - p_0 > \tau$, he has all the information from $V_0$.
  Since the last message from $V_1$ (for honest computation happens at time $t = \tau$) will reach $A_0$ at $t = \tau + p_P - p_0$, he also has all the information from $V_1$.
  By symmetry, $A_1$ also has enough time to gather all the information.
  Finally, because they pre-share the identical randomness tape, their response will be consistent as a single classical prover.
  Therefore, this adversarial strategy has success probability $w$, and the only setup they need is some pre-shared classical randomness.
  By soundness of the protocol, $w \le s + \negl(\lambda)$.
\end{proof}

We remark that these requirements except classical verifiability are rather minimal, and all known position verification protocols, including ours, satisfy these properties.
Furthermore, the position verification protocol could be arbitrarily many rounds.

\fi

\ifitcs\else\printbibliography\fi

\ifexabs\else
\appendix

\section{NTCFs and \oneoftwo Puzzles}
\label{sec:NTCF_prelim}

The following definition of NTCF families is taken verbatim from \cite[Definition 6]{brakerski2018cryptographic}.
For a more detailed exposition of the definition, we refer the readers to the prior work.

\begin{definition}[NTCF family]\label{def:trapdoorclawfree}
Let $\lambda$ be a security parameter. Let $\sX$ and $\sY$ be finite sets.
 Let $\mathcal{K}_{\mathcal{F}}$ be a finite set of keys. A family of functions 
$$\mathcal{F} = \big\{f_{k,b} : \sX\rightarrow \mathcal{D}_{\sY} \big\}_{k\in \mathcal{K}_{\mathcal{F}},b\in\{0,1\}}$$
is called a \emph{noisy trapdoor claw free (NTCF) family} if the following conditions hold:

\begin{enumerate}
\item{\textbf{Efficient Function Generation.}} There exists an efficient probabilistic algorithm $\textrm{GEN}_{\mathcal{F}}$ which generates a key $k\in \mathcal{K}_{\mathcal{F}}$ together with a trapdoor $t_k$: 
$$(k,t_k) \leftarrow \textrm{GEN}_{\mathcal{F}}(1^\lambda).$$
\item{\textbf{Trapdoor Injective Pair.}} For all keys $k\in \mathcal{K}_{\mathcal{F}}$ the following conditions hold. 
\begin{enumerate}
\item \textit{Trapdoor}: There exists an efficient deterministic algorithm $\textrm{INV}_{\mathcal{F}}$ such that for all $b\in \{0,1\}$,  $x\in \sX$ and $y\in \supp(f_{k,b}(x))$, $\textrm{INV}_{\mathcal{F}}(t_k,b,y) = x$. Note that this implies that for all $b\in\{0,1\}$ and $x\neq x' \in \sX$, $\supp(f_{k,b}(x))\cap \supp(f_{k,b}(x')) = \emptyset$. 
\item \textit{Injective pair}: There exists a perfect matching $\sR_k \subseteq \sX \times \sX$ such that $f_{k,0}(x_0) = f_{k,1}(x_1)$ if and only if $(x_0,x_1)\in \sR_k$. \end{enumerate}

\item{\textbf{Efficient Range Superposition.}}
For all keys $k\in \mathcal{K}_{\mathcal{F}}$ and $b\in \{0,1\}$ there exists a function $f'_{k,b}:\sX\to \mathcal{D}_{\sY}$ such that the following hold.
\begin{enumerate} 
\item For all $(x_0,x_1)\in \mathcal{R}_k$ and $y\in \supp(f'_{k,b}(x_b))$, INV$_{\mathcal{F}}(t_k,b,y) = x_b$ and INV$_{\mathcal{F}}(t_k,b\oplus 1,y) = x_{b\oplus 1}$. 
\item There exists an efficient deterministic procedure CHK$_{\mathcal{F}}$ that, on input $k$, $b\in \{0,1\}$, $x\in \sX$ and $y\in \sY$, returns $1$ if  $y\in \supp(f'_{k,b}(x))$ and $0$ otherwise. Note that CHK$_{\mathcal{F}}$ is not provided the trapdoor $t_k$.
\item For every $k$ and $b\in\{0,1\}$,
$$ \E_{x\leftarrow_U \sX} \big[H^2(f_{k,b}(x),f'_{k,b}(x))\big] \leq \mu(\lambda),$$
 for some negligible function $\mu(\cdot)$. Here $H^2$ is the Hellinger distance. Moreover, there exists an efficient procedure  SAMP$_{\mathcal{F}}$ that on input $k$ and $b\in\{0,1\}$ prepares the state
\begin{equation*}
    \frac{1}{\sqrt{|\sX|}}\sum_{x\in \sX,y\in \sY}\sqrt{(f'_{k,b}(x))(y)}\ket{x}\ket{y}.
\end{equation*}

\end{enumerate}

\item{\textbf{Adaptive Hardcore Bit.}}
For all keys $k\in \mathcal{K}_{\mathcal{F}}$ the following conditions hold, for some integer $w$ that is a polynomially bounded function of $\lambda$. 
\begin{enumerate}
\item For all $b\in \{0,1\}$ and $x\in \sX$, there exists a set $G_{k,b,x}\subseteq \{0,1\}^{w}$ such that $\Pr_{d\leftarrow_U \{0,1\}^w}[d \notin G_{k,b,x}]$ is negligible, and moreover there exists an efficient algorithm that checks for membership in $G_{k,b,x}$ given $k,b,x$ and the trapdoor $t_k$. 
\item There is an efficiently computable injection $\mathcal{J}:\sX\to \{0,1\}^w$, such that $\mathcal{J}$ can be inverted efficiently on its range, and such that the following holds. If
\begin{align*}\label{eq:defsetsH}
H_k &= \big\{(b,x_b,d,d\cdot(\mathcal{J}(x_0)\oplus \mathcal{J}(x_1)))\,|\; b\in \{0,1\}, (x_0,x_1)\in \mathcal{R}_k, d\in G_{k,0,x_0}\cap G_{k,1,x_1}\big\},%
\\
\overline{H}_k &= \{(b,x_b,d,c)\,|\; (b,x,d,c\oplus 1) \in H_k\big\},
\end{align*}
then for any quantum polynomial-time procedure $\mathcal{A}$ there exists a negligible function $\mu(\cdot)$ such that 
\begin{equation*}\label{eq:adaptive-hardcore}
\left|\Pr_{(k,t_k)\leftarrow \textrm{GEN}_{\mathcal{F}}(1^{\lambda})}[\mathcal{A}(k) \in H_k] - \Pr_{(k,t_k)\leftarrow \textrm{GEN}_{\mathcal{F}}(1^{\lambda})}[\mathcal{A}(k) \in\overline{H}_k]\right| \leq \mu(\lambda).
\end{equation*}
\end{enumerate}

\end{enumerate}
\end{definition}

\begin{theorem}[{\cite[Theorem 4.1]{brakerski2021cryptographic}}]
  \label{thm:lwe-ntcf}
  Assuming the post-quantum hardness of Learning with Errors (LWE) problem, NTCF families exist.
\end{theorem}

\begin{construction}
  \label{construct:ntcf21o2}
  Let $\mathcal F$ be an NTCF.
  An 1-of-2 puzzle can be constructed as follows.
\begin{itemize}
\item The $\keygen$ algorithm in 1-of-2 puzzle generates a public key $k$ and a secret key (trapdoor) $t_k$ for NTCF $f_k = \{f_{k,b}: \mathcal{X} \to \mathcal{D}_\mathcal{Y} \}_{k \in \mathcal{K}_\mathcal{F}, b \in \{0,1\}}$. Let $\pk = k$ and $\sk = (k, t_k)$.

\item The $\obligate$ algorithm prepares the evaluation of $f'_k$ (instead of $f_k$) on the uniform superposition over all inputs, measures the image register to obtain $y$ and apply $\mathcal{J}$ on the input.
Let $\rho$ be the post-measurement state of the input register, where $\rho \approx \ket\psi\bra\psi$, and
$$\ket\psi = \frac{1}{\sqrt{2}}(\ket{0}\ket{\mathcal{J}(x_0)} +\ket{1}\ket{\mathcal{J}(x_1)})\ket{y}.$$%
This is because $f_{k}$ has the trapdoor injective pair property and efficient range superposition property of $f'$; $\mathcal{J}$ is an injective procedure. 

\item $\solve(\pk, y, \rho, 0)$ corresponds to measuring the two registers in state $\rho$ in computational basis and outputting $\ans = (b, \mathcal{J}(x_b)), b \in \{0,1\}, x_b \in \{x_0, x_1\}$.

$\ver(\sk,y,0,\ans)$ takes in secret key $\sk = (k, t_k)$, an obligation value $y$, the challenge bit $b = 0$ and $\ans$. In this case $\ans = (b, v)$. The verification algorithm first obtains $x' = \mathcal{J}^{-1}(v)$ and outputs $1$ if and only if $y \in \supp(f_{k,b}'(x))$. 

\item $\solve(\pk, y, \rho, 1)$ corresponds to   measuring state in Hadamard basis and outputting the measurement result $\ans = (b,d)$. %

$\ver(\sk,y,1,\ans)$ takes in secret key $\sk = (k, t_k)$, an obligation value $y$, the challenge bit $b = 1$ and $\ans$.
In this case, $\ans = (b, d)$. It outputs $1$ if and only if $d \ne 0$ and $d \cdot (\mathcal{J}(x_0) + \mathcal{J}(x_1)) = b$. 
\end{itemize}
\end{construction}

We remark that \Cref{construct:ntcf21o2} differs from \cite[Algorithm 1]{radian2020semi} where we move the evaluation of $\mathcal J$ from $\solve$ to $\obligate$.
Since evaluating $\mathcal J$ only acts on $\rho$ and is independent of $b$, this is only a conceptual change and all the properties of the \oneoftwo puzzle preserve.

We can deduce the following facts by staring at the construction.

\begin{fact}\label{cor:XZ_measure}
\Cref{construct:ntcf21o2} has an XZ-solver. 
\end{fact}

\begin{fact}\label{cor:zero_pub_ver}
\Cref{construct:ntcf21o2} has $0$-challenge-public-verifiability.
\end{fact}
\begin{proof}
  By NTCF definition, we can use procedure ${\sf CHK}_{\mathcal{F}}$ to check if $y \in \supp(f_{k,b}'(x))$ using only $\pk = k$.
\end{proof}

\fi

\end{document}